\documentclass[12pt, draftclsnofoot, onecolumn]{IEEEtran}

\usepackage{graphicx} 

\usepackage{footnote}
\usepackage{amsmath}

\usepackage{mathtools}

\usepackage{amssymb} 
\usepackage{enumerate}
\usepackage{array}
\usepackage{cite}
\usepackage{gensymb}

\usepackage{balance}

\usepackage{tikz}
\usetikzlibrary{calc,shapes,snakes,angles,positioning,intersections,quotes,decorations.markings}
\usepackage{tkz-euclide}
\usetkzobj{all}
\def\myrad{2.5cm}

\usepackage{pgfplots}
\usepackage{tikz}
\usetikzlibrary{calc,shapes,snakes,angles,positioning,intersections,quotes,decorations.markings}
\usepackage{tkz-euclide}
\usetkzobj{all}
\pgfplotsset{compat=1.11}

\usepackage{verbatim}

\newcommand{\matr}[1]{\mathbf{#1}}
\newcommand{\E}{\mathrm{E}}

\DeclarePairedDelimiterX\setc[2]{\{}{\}}{\,#1 \;\delimsize\vert\; #2\,} 

\makeatletter
\newsavebox{\mybox}\newsavebox{\mysim}
\newcommand{\distras}[1]{%
  \savebox{\mybox}{\hbox{\kern3pt$\scriptstyle#1$\kern3pt}}%
  \savebox{\mysim}{\hbox{$\sim$}}%
  \mathbin{\overset{#1}{\kern\z@\resizebox{\wd\mybox}{\ht\mysim}{$\sim$}}}%
}

\DeclarePairedDelimiter\floor{\lfloor}{\rfloor}

\usepackage{amsthm}

\newtheorem{theorem}{Theorem}
\newtheorem{corollary}{Corollary}[theorem]
\newtheorem{lemma}[theorem]{Lemma}
\newtheorem{proposition}[theorem]{Proposition}

\theoremstyle{remark}
\newtheorem*{remark}{Remark}

\theoremstyle{definition}
\newtheorem{definition}{Definition}

\theoremstyle{definition}
\newtheorem{assumption}{Assumption}

\theoremstyle{definition}
\newtheorem{approximation}{Approximation}

\usepackage[cmintegrals]{newtxmath}

\usepackage{url}

\begin{document}


\title{A Statistical Characterization of Localization \\ Performance in Wireless Networks}
%


\author{Christopher~E.~O'Lone,~\IEEEmembership{Graduate~Student~Member,~IEEE,}
		Harpreet~S.~Dhillon,~\IEEEmembership{Member,~IEEE,} 
		and~R.~M.~Buehrer,~\IEEEmembership{Fellow,~IEEE}
\thanks{The authors are with Wireless@VT, Bradley Department of Electrical and Computer Engineering, Virginia Tech, Blacksburg, VA, USA. Email: \{olone, hdhillon, buehrer\}@vt.edu.}
\thanks{This paper was presented in part at the IEEE ICC 2017 Workshop on Advances in Network Localization and Navigation (ANLN), Paris, France \cite{ICC_Paper}.}}

\maketitle

\vspace{-2pt}
\begin{abstract}
	
	 Localization performance in wireless networks has been traditionally benchmarked using the Cram\'{e}r-Rao lower bound (CRLB), given a \emph{fixed} geometry of anchor nodes and a target.  However, by endowing the target and anchor locations with distributions, this paper recasts this traditional, scalar benchmark as a random variable.  The goal of this work is to derive an analytical expression for the distribution of this now \emph{random} CRLB, in the context of Time-of-Arrival-based positioning.  
	
	To derive this distribution, this work first analyzes how the CRLB is affected by the order statistics of the angles \emph{between} consecutive participating anchors (\emph{i.e.}, internodal angles).  This analysis reveals an intimate connection between the second largest internodal angle and the CRLB, which leads to an accurate approximation of the CRLB.  Using this approximation, a closed-form expression for the distribution of the CRLB, \emph{conditioned} on the number of participating anchors, is obtained.  
	
	Next, this conditioning is eliminated to derive an analytical expression for the \emph{marginal} CRLB distribution.  Since this marginal distribution accounts for all target and anchor positions, across all numbers of participating anchors, it therefore statistically characterizes localization error throughout an \emph{entire} wireless network.  This paper concludes with a comprehensive analysis of this new network-wide-CRLB paradigm.

\end{abstract}

\vspace{-1pt}



\section{Introduction}

	\IEEEPARstart{T}{he} Global Positioning System (GPS) has for decades been the standard mechanism for position location anywhere in the world.  However, the deployment locations of recent and emerging wireless networks have begun to put a strain on the effectiveness of GPS as a localization solution.  For example, as populations increase, precipitating the expansion of urban environments, cell phone use in urban canyons, as well as indoors, is continually increasing.  The rise of these GPS-constrained environments highlight the need to fall back on the existing network infrastructure for localization purposes.

	Additionally, with the emergence of Wireless Sensor Networks (WSNs) and their increased emphasis on energy efficiency and cost-effectiveness, \cite{WSN_Survey}, \cite{WSN_Devices}, the possibility of equipping each potential target node with a GPS chip quickly becomes impractical.  Furthermore, the deployment of these networks in GPS-constrained environments further necessitates a reliance on the terrestrial network for a localization solution.  Thus, localization within a network, performed by the network itself in the absence of GPS, has begun to garner attention, \cite{Locating_the_Nodes}, \cite{Localization_WSNs}.
	
	Benchmarking localization performance in wireless networks has traditionally been done using the Cram\'{e}r-Rao lower bound, which provides a lower bound on the position error of \emph{any} unbiased estimator \cite{Kay}.  Common practice has been to analyze the CRLB in \emph{fixed} scenarios of anchor nodes and a target, \emph{e.g.}, \cite{Fixed_Setup_1}, \cite{Fixed_Setup_2}, \cite{Fixed_Setup_3}, \cite{Fixed_Setup_4}.  This strategy produces a scalar/fixed value for the CRLB and is specific to the scenario being analyzed.  While this idea does provide insight into fundamental limits of localization performance, it is rather limited in that it does not take into account all possible setups of anchor nodes and target positions within a network.
	
	In order to account for all possible setups, it is useful to appeal to the field of stochastic geometry.  Whereas in the past stochastic geometry has been applied towards the study of ``connectivity, capacity, outage probability, and other fundamental limits of wireless networks,'' \cite{StochasticGeometryPaper}, \cite{StochasticGeometry}, we now however apply it towards the study localization performance in wireless networks.  Modeling anchor node and target placements with point processes opens the possibility of characterizing the CRLB over \emph{all} setups of anchor nodes and target positions.  Thus, the CRLB is no longer a fixed value, but rather a random variable (RV) conditioned on the number of participating anchors, where the randomness is induced by the inherent randomness of the anchor positions.  Upon marginalizing out the number of participating anchors, the resulting marginal distribution of the CRLB will characterize localization performance throughout an entire wireless network.

\subsection{Related Work} \label{RelatedWork}
	The quest for a network-wide distribution of localization performance comprises two main steps.  The first step involves finding the distribution of the CRLB \emph{conditioned} on the number of participating anchor nodes, and the second step involves finding the probability that a given number of anchors can participate in a localization procedure.   
	
	With regards to the first step, there have been several attempts in the literature to obtain this conditional distribution.  An excellent first attempt can be found in the series of papers, \cite{RSS}, \cite{TOA}, \cite{RSS_TOA_AOA}, in which approximations of this conditional distribution were presented for Received-Signal-Strength (RSS), Time-of-Arrival (TOA), and Angle-of-Arrival (AOA) based localization, respectively.  These approximations were obtained through asymptotic arguments by driving the number of participating anchor nodes to infinity.  While these approximate distributions are very accurate for larger numbers of participating anchors, they are less than ideal for lower numbers.  However, it is desirable that this conditional distribution be accurate for lower numbers of participating anchors, since this is the dominate case in terrestrial networks, \emph{e.g.}, cellular. 
	
	This conditional distribution of the CRLB was also explored in \cite{Fengyu_Zhou}.  Here, the authors were able to derive the true expression for this conditional distribution through a clever re-writing of the CRLB using complex exponentials.  This distribution was then used to derive and analyze the so-called ``localization outage probability'' in scenarios with a \emph{fixed} number of randomly placed anchor nodes. While this expression represents the true conditional distribution, its complexity puts it at a disadvantage over simpler approximations (discussed further in Section \ref{Cond_CRLB_Dist}).
		
	The second step, which involves finding the participation probability of a given number of anchor nodes, was explored in \cite{Javier_Journal}.  In this work, the authors modeled a cellular network with a homogeneous Poisson point process (PPP), which consequently allowed them to derive bounds on \emph{L-localizability}, \emph{i.e.}, the probability that a mobile device can hear \emph{at least} $L$ base stations for participation in a localization procedure.\footnote{The term \emph{hearability} (or \emph{hearable}) is used to describe anchor nodes that are able to participate in a localization procedure, \emph{i.e.}, their received SINR at the target is above some threshold.}  Further, by employing a ``dominant interferer analysis,'' they were able to derive an accurate expression for the probability of \emph{L-localizability}, which can easily be extended to give the probability of hearing a given number of anchor nodes for participation in a localization procedure.

\subsection{Contributions} \label{ContributionsSection}

	This paper proposes a novel statistical characterization of a wireless network's ability to perform TOA-based localization. Using stochastic geometry to model target and anchor node placements throughout a network, and using the CRLB as the localization performance benchmark, this paper presents an analytical derivation of the CRLB's \emph{distribution}.\footnote{We will be using the square root of the CRLB as our performance benchmark, however, we just state the CRLB here as to not unnecessarily clutter the discussion.  This will be described further in Section \ref{LPB}.}  
This distribution offers many insights into localization performance within wireless networks that previously were only attainable through lengthy, parameter-specific, network simulations.  
Thus, this distribution
\begin{enumerate}
\item offers a means for comparing networks in terms of their localization performance, by enabling the calculation of network-wide localization statistics, \emph{e.g.}, avg. localization error,
\item \label{network_parameters} unlocks insight into how changing network parameters, such as SINR thresholds, processing gain, frequency reuse, etc., affect localization performance throughout a network, and
\item provides network designers with an analytical tool for determining whether a network meets, for example, the FCC E911 standards \cite{E911}.
\end{enumerate}

	In pursuit of this distribution, this paper makes four key contributions.  First, this work presents an analysis of how the CRLB is affected by the order statistics of the internodal angles.  This analysis reveals an intimate connection between the second largest internodal angle and the CRLB, which leads to an accurate approximation of the CRLB.  Second, this approximation is then used to obtain the distribution of the CRLB \emph{conditioned} on the number of hearable anchors.  Although this is a distribution of the CRLB approximation, its simplicity and accuracy offer clear advantages over the true distribution presented in \cite{Fengyu_Zhou} and the approximate distribution presented in \cite{TOA}.  These advantages are discussed further in Section \ref{Cond_CRLB_Dist}.	
	

	Third, this work then takes a major step by combining this conditional distribution of the CRLB with the distribution of the number of participating anchors.  This eliminates the conditioning on a given number of anchors, allowing for the analytical expression of the \emph{marginal} CRLB distribution to be obtained. Since this marginal distribution now simultaneously accounts for all possible target and anchor node positions, across all numbers of participating anchor nodes, it therefore statistically characterizes localization error throughout an \emph{entire} wireless network, and thus signals a departure from the existing literature.  Additionally, since the two component distributions are parameterized by various network parameters, our resulting marginal distribution of the CRLB will also be parameterized by these network parameters.  Consequently, our final contribution involves a comprehensive analysis of this new network-wide CRLB paradigm, where we examine how varying network parameters affects the distribution of the CRLB --- thereby revealing how these network parameters affect localization performance throughout the network.


\section{Problem Setup}
	This section details the assumptions we propose in determining the network layout as well as the localization procedure.  Additionally, we describe important notation and definitions used throughout the paper and conclude with how the assumptions impact the network setup. 

\subsection{Network Setup and Localization Assumptions}
\begin{assumption} \label{HPPP}
We assume a ubiquitous, two-dimensional wireless network with anchor nodes distributed according to a homogeneous PPP over $\mathbb{R}^2$.  Further, we assume potential targets to be distributed likewise, where the anchor and target point processes are assumed to be independent.
\end{assumption}
\begin{remark}
This assumption for modeling wireless networks is common in the literature, \emph{e.g.}, \cite{CoverageRate}, \cite{Modeling}, \cite{PoissonProcesses}, \cite{PoissonSignals}, \cite{StrongShadowing}.
\end{remark}

\begin{assumption} \label{2WTOA}
We assume a 2-Way TOA-based positioning technique is used within the two-dimensional network.
\end{assumption}

\begin{remark}
Although Time-Difference-of-Arrival (TDOA) is also commonly implemented, 2-Way TOA represents a viable approach, and additionally offers a lower bound on TDOA performance.  Furthermore, the 2-Way assumption eliminates the need for clock synchronization between the target and anchor nodes.
\end{remark}

\begin{assumption} \label{Independent_Range}
Range measurements are independent and exhibit zero-mean, normally distributed range error.
\end{assumption}

\begin{remark}
A reader familiar with wireless positioning will recognize this as a classic Line-of-Sight (LOS) assumption.  However, those familiar with localization in terrestrial networks will realize that Non-Line-of-Sight (NLOS) measurements are more common.  Thus, while we move forward under this LOS assumption in order to make progress under this new paradigm, we will see in Section \ref{NumericalAnalysis} that we can adapt our model to accommodate NLOS measurements by selecting ranging errors consistent with NLOS propagation. 
\end{remark}

\begin{assumption} \label{Common_Error}
The range error variance, $\sigma_r^2$, is common among measurements from participating anchor nodes and is considered a known quantity.
\end{assumption}
\begin{remark}
This assumption is often made in the literature for range-based localization, \emph{e.g.}, \cite{Fixed_Setup_4}, \cite{PEB1}, and although not realistic in every scenario, it allows us to gain initial insight into the problem and will be relaxed in future work.
\end{remark}

\subsection{Notation}

	The notation used throughout this paper can be found in Table \ref{Notation}. 

\begin{table}[t]
\renewcommand{\arraystretch}{0.72}
\caption{Summary of Notation} 
\label{Notation}
\vspace{-7pt}
\centering 
\begin{tabular}{c l | c l } 
\hline\hline 
\bfseries Symbol & \bfseries Description & \bfseries Symbol & \bfseries Description \\ [0.3ex]
\hline  
$f_X(\cdot)$     & Probability distribution function (pdf) of $X$ & $F_X(\cdot)$   	& Cumulative distribution function (cdf) of $X$\\ 
$h(X)$    		& Differential entropy of $X$  &    $I(X;Y)$ 	& Mutual Information between $X$ and $Y$ \\
$u(x)$                              	 & Unit step function, 0 when $x<0$, 1 when $x\geq0$ & $\mathcal{N}(\mu, \sigma^2)$  & Normal distribution, mean $\mu$, variance $\sigma^2$ \\
$\lVert \cdot \rVert$  & $\ell_2$-norm & $ \text{P}[A]$ & Probability of event $A$\\
$\mathbf{1}[\cdot]$	& Indicator function; $\mathbf{1}[A] = 1$ if $A$ true, 0 otherwise & $\floor*{\cdot}$ & Floor function\\
$\text{tr}(\matr{X})$	& Trace of the matrix $\matr{X}$ & $\matr{X}^T$	& Transpose of the matrix $\matr{X}$ \\
$L$     	& Number of participating anchors, $L\smashoperator\in\{0, 1, 2, \dots\}$  & $N$		& Max anchors tasked to perform localization\\
$S$                                      	& The localization performance benchmark & $M$ & Fixed localization error, in meters, $M \in \mathbb{R}^+$ \\
$\Theta_k$    & Anchor node angle $k$, a random variable & $\theta_k$ &  A realization of $\Theta_k$ \\
$\angle_k$    & Internodal angle $k$, a random variable  & $\varphi_k$ & A realization of $\angle_k$  \\
$X_{(k)}$                   	   & $k^\text{th}$ order statistic of a sequence $\{X_i\}$ of RVs & $\lambda$ & Density of PPP of anchor locations \\
$\alpha$       & Path loss exponent, $\alpha > 2$ & $\gamma$ & Processing gain\\
$\beta$                      	   & Post-proc. SINR threshold & $q$  & Average network load (\emph{i.e.}, network traffic) \\
$\mathcal{K}$ 	 & Frequency reuse factor & $\sigma_r^2$ & Common range error variance\\
\hline 
\end{tabular}
\end{table}

\subsection{Localizability}
	We now define the terms \emph{localizable} and \emph{unlocalizable},  which were introduced in \cite{Javier_Journal}.
\begin{definition} \label{localizable}
We say that a target is \emph{localizable} if it detects localization signals from a sufficient number of anchor nodes such that its position can be determined without ambiguity.  
\end{definition}
\begin{remark}
Under Assumption \ref{2WTOA}, this implies that $L \geq 3$.  We also define \emph{unlocalizable} to be the negation of Definition \ref{localizable}.
\end{remark}
\emph{For the purposes of this setup and subsequent derivations, we will initially only consider scenarios in which the target is localizable, to avoid unnecessary complication.  Later in Section \ref{JaviersWork} and those which follow, we will account for scenarios in which the target is unlocalizable, and will modify our results accordingly.}

\subsection{Impact of Assumptions} \label{Impact_of_Assumptions}
	With these assumptions in place, we now describe how they impact the network setup. From Assumption \ref{HPPP}, since the anchors and potential targets are distributed by independent, homogeneous PPPs over $\mathbb{R}^2$ (\emph{i.e.}, stationary), then without loss of generality, we may perform our analysis for a typical target placed at the origin of the $xy$-plane \cite{PoissonProcesses}.  This is due to the fact that the independence and stationarity assumptions imply that no matter where the target is placed in the network, the distribution of anchors relative to the target appears the same.

	Next, we assume that the number of hearable anchors is some \emph{fixed} value, $L$, and begin by numbering these anchors in terms of their increasing distance from the origin (target position).  This is depicted in Fig. \ref{Setup1} for a \emph{particular realization} of a homogeneous PPP in which there are four hearable anchor nodes.  Fig. \ref{Setup1} also depicts how their corresponding angles, measured counterclockwise from the $+x$-axis, are labeled accordingly.   Assumption \ref{HPPP} further implies that these angles of the hearable anchors are i.i.d. random variables that come from a uniform distribution on [0, 2$\pi$).
\begin{definition} \label{ANA}
If the target is placed at the origin of an $xy$-plane, then the term \emph{anchor node angle}, $\Theta_k$, is defined to be the angle corresponding to hearable anchor node $k$, measured counterclockwise from the $+x$-axis.  Note that $\Theta_k\distras{\text{i.i.d.}}\text{unif}[0,2\pi)$, $\forall k\in\{1, \dots, L\}$.
\end{definition}

\begin{figure}[t]
\centering
\begin{minipage}[t]{0.47\textwidth}
\centering
\begin{tikzpicture} [scale=0.7]
\coordinate (O) at (0,0);

\draw [<->,thick] (0,-4)--(0,4) node (yaxis) [above] {$y$};
\draw [<->,thick] (-4,0)--(4,0) node (xaxis) [right] {$x$};
  
\pgfmathsetmacro{\angleA}{35}
\draw (O) --(\angleA:3) coordinate (A);

\pgfmathsetmacro{\angleB}{55}
\draw (O) --(\angleB:2) coordinate (B);

\pgfmathsetmacro{\angleC}{135}
\draw (O) --(\angleC:1.5) coordinate (C);

\pgfmathsetmacro{\angleD}{280}
\draw (O) --(\angleD:2.5) coordinate (D);


\pgfmathsetmacro{\angleE}{340}
\coordinate (O) --(\angleE:3.5) coordinate (E);

\pgfmathsetmacro{\angleF}{97}
\coordinate (O) --(\angleF:4) coordinate (F);

\pgfmathsetmacro{\angleG}{190}
\coordinate (O) --(\angleG:4.5) coordinate (G);

\pgfmathsetmacro{\angleH}{240}
\coordinate (O) --(\angleH:4) coordinate (H);

\fill (A) circle[radius=3pt] node[right] {$4$};
\fill (B) circle[radius=3pt] node[right] {$2$}; 
\fill (C) circle[radius=3pt] node[above,left] {$1$}; 
\fill (D) circle[radius=3pt] node[below] {$3$}; 

\fill (E) circle[radius=3pt] node[right] {};
\fill (F) circle[radius=3pt] node[right]{};
\fill (G) circle[radius=3pt] node[above,left]{};
\fill (H) circle[radius=3pt] node[below] {};


\pgfmathsetmacro{\angleX}{0}
\coordinate (X) --(\angleX:3) coordinate (X);

\draw pic [draw,dashed,angle radius=2cm, pic text=$\theta_{4}$,  angle eccentricity=1.15] {angle = X--O--A};

\draw pic [draw,dashed,angle radius=1.3cm, pic text=$\theta_{2}$,  angle eccentricity=1.2] {angle = X--O--B};

\draw pic [draw,dashed,angle radius=0.8cm, pic text=$\theta_{1}$,  angle eccentricity=1.3] {angle = X--O--C};

\draw pic [draw,dashed,angle radius=0.4cm,  angle eccentricity=1.6] {angle = X--O--D};
\draw (-0.8,0.8) node[below]{$\theta_{3}$};

\end{tikzpicture}
\caption{\textsc{Initial Labeling Scheme.} The dots represent a particular realization of anchors placed according to a homogeneous PPP.  The origin represents the location of the target.  The anchors participating in the localization procedure are labeled in increasing order w.r.t. their distance from the origin.  Their corresponding anchor node angles are labeled accordingly. Note the realization of the RV $\Theta_k$ is $\Theta_k = \theta_k$.}
\label{Setup1}
\end{minipage}\hfill
\begin{minipage}[t]{0.47\textwidth}
\centering
\begin{tikzpicture} [scale=0.7]
\coordinate (O) at (0,0);

\draw [<->,thick] (0,-4)--(0,4) node (yaxis) [above] {$y$};
\draw [<->,thick] (-4,0)--(4,0) node (xaxis) [right] {$x$};
  
\pgfmathsetmacro{\angleA}{35}
\draw (O) --(\angleA:\myrad) coordinate (A);

\pgfmathsetmacro{\angleB}{55}
\draw (O) --(\angleB:\myrad) coordinate (B);

\pgfmathsetmacro{\angleC}{135}
\draw (O) --(\angleC:\myrad) coordinate (C);

\pgfmathsetmacro{\angleD}{280}
\draw (O) --(\angleD:\myrad) coordinate (D);

\fill (A) circle[radius=3pt] node[right] {$\theta_{(1)}$};
\fill (B) circle[radius=3pt] node[right] {$\theta_{(2)}$}; 
\fill (C) circle[radius=3pt] node[above,left] {$\theta_{(3)}$}; 
\fill (D) circle[radius=3pt] node[below] {$\theta_{(4)}$};

\draw (O) node[circle,inner sep=1.5pt] {} circle [radius=\myrad];

\draw pic [draw,-,angle radius=1cm, pic text=$\varphi_{(1)}$,  angle eccentricity=1.38] {angle = A--O--B};
  
\draw pic [draw,-,angle radius=0.8cm,  angle eccentricity=1.3] {angle = B--O--C};
\draw (-0.35,1.8) node[below]{$\varphi_{(2)}$};

\draw pic [draw,-,angle radius=0.5cm, pic text=$\varphi_{\text{(4)}}$,  angle eccentricity=1.8] {angle = C--O--D};

\draw pic [draw,-,angle radius=0.6cm, pic text=$\varphi_{(3)}$,  angle eccentricity=1.6] {angle = D--O--A};

%
%
%
%

\end{tikzpicture}
\vspace{-3pt}
\caption{\textsc{Equivalent Setup.}  This is the realization from Fig. \ref{Setup1}, where only participating anchor nodes and the internodal angles they trace out are considered, whereas their distances from the target are not.  The realizations of the RVs are given by $\Theta_{(k)} = \theta_{(k)}$ and $\angle_{(k)} = \varphi_{(k)}$. Note, the anchor node angle order stats ``renumber'' the participating anchors in terms of increasing angle c.c.w. starting from the $+x$-axis.}
\label{Setup2}

\end{minipage}
\end{figure}

	In later sections we will see that the distances between the anchor nodes and the target are important for determining how many anchors are able to participate in a given localization procedure.  However, under Assumption \ref{Common_Error}, once particular anchor nodes are identified as participating in a localization procedure, then they are endowed with the common range error variance.  
As we will see in the following section, this assumption, along with Assumptions \ref{2WTOA} and \ref{Independent_Range} will lead to the CRLB (with $L$ fixed) as being dependent on \emph{only} the angles \emph{between} participating anchor nodes (\emph{i.e.}, internodal angles).  Thus, since the CRLB expression is only dependent on the internodal angles, the distances between \emph{participating} anchor nodes and the target need not be considered.  Hence, we may view the participating anchors as being placed on a circle about the origin.  This is depicted in Fig. \ref{Setup2}, under the same PPP realization as in Fig \ref{Setup1}.

	Next, we formally define the term \emph{internodal angle}.  Since the anchor node angles are such that $\Theta_k\distras{\text{i.i.d.}}\text{unif}[0,2\pi)$, for $k \in \{1, \dots, L \}$, we may examine their corresponding order statistics, $\Theta_{(1)}, \Theta_{(2)}, \ldots,  \Theta_{(L)}$, where $0 \leq \Theta_{(1)} \leq \Theta_{(2)} \leq \dots \leq \Theta_{(L)} \leq 2\pi$ by definition.  Thus, the order statistics of the participating anchor node angles effectively ``renumber'' the nodes in terms of increasing angle, starting counterclockwise from the $+x$-axis.  This is also depicted in Fig. \ref{Setup2}.
\begin{definition} \label{Internodal_Angles}
If participating anchor nodes are considered according to their anchor node angle order statistics, then an \emph{internodal angle}, $\angle_k$, is defined to be the angle between two consecutive participating anchor nodes.  That is, $\angle_k = \Theta_{(k+1)} - \Theta_{(k)}$ for $1 \leq k < L$ and $\angle_{L} = 2\pi - (\Theta_{(L)} - \Theta_{(1)})$.
\end{definition}
\begin{remark}
Since the internodal angles are functions of RVs, they themselves are RVs. Thus, we may also consider their order statistics, $\angle_{(1)}, \angle_{(2)}, \dots, \angle_{(L)}$.  These order statistics of the internodal angles are depicted for the particular PPP realization in Fig. \ref{Setup2}.
\end{remark}

	In summary, Fig \ref{Setup2} depicts an example of a typical setup realization, given $L = 4$, once all of the assumptions are taken into consideration.


\section{Derivation of the Network-Wide CRLB Distribution} \label{NWCD}


	In this section, we first formally define our localization performance benchmark: the square root of the CRLB.
Using this definition and assuming a random placement of anchors, as well as a random $L$, we then describe how this work generalizes localization performance results currently in the literature. In what follows, we present the steps necessary to derive the marginal distribution of our localization performance benchmark.

\vspace{-5pt}
\subsection{The Localization Performance Benchmark} \label{LPB}
	
	Consider the traditional localization scenario, where the number of participating anchor nodes ($L$) and their positions, as well as the target position, are all \emph{fixed}.
We represent the set of coordinates of these anchors by
\vspace{-7pt}
\begin{align*}
\Psi_L = \setc[\Big]{ \psi_i \in \mathbb{R}^2 } { \psi_i = \big[ x_i , y_i \big]^T, ~i \in \{1, 2, 3, ... , L\} }.
\end{align*}
The coordinates of the target are denoted by $\psi_t = \big[ x_t , y_t \big]^T$.

	Next, under Assumptions \ref{2WTOA}, \ref{Independent_Range}, and \ref{Common_Error}, the 2-Way range measurements between the target and the $L$ participating anchors are given by
\vspace{-12pt}
\begin{align*}
r_i = d_i + n_i ~,
\end{align*}
\vspace{-4pt}
where $r_i$ is the measured distance divided by 2 (\emph{i.e.}, the measured 1-Way distance),
\begin{align*} 
d_i = \lVert \psi_i - \psi_t \rVert = \sqrt{\big(x_i - x_t\big)^2 + \big(y_i - y_t\big)^2}
\end{align*}
is the true 1-Way distance, and $n_i \distras{\text{i.i.d.}} \mathcal{N}(0, \sigma_r^2)$ for all $i \in \{1, 2, \dots, L\}$.

\begin{remark}
Note that under Assumption \ref{Common_Error}, $\sigma_r^2$ is common among the range measurements.  For 2-Way TOA we set this to be $\sigma_r^2 = \sigma_\text{2-Way}^2 / 4$.  If 1-Way TOA is considered, then $\sigma_r^2 = \sigma_\text{1-Way}^2$.  Thus, moving forward we may consider the range measurements to always be 1-Way, regardless of whether 1-Way or 2-Way TOA is used.
\end{remark}
	
	Continuing, Assumption \ref{Independent_Range} enables the likelihood function to be easily written as a product.  Denoting the vector of range measurements as $\matr{r} = [r_1, \dots, r_L]^T$, the likelihood function is
\begin{align*}
\mathcal{L}\Big(\psi_t \, \Big| \, \matr{r}, \Psi_L, \sigma_r^2 \Big) = \prod_{i = 1}^{L} \frac{1}{\sqrt{2\pi} \sigma_r} \exp{ \Bigg( \! -\frac{(r_i-d_i)^2}{2\sigma_r^2} \Bigg)}.
\end{align*}
From this likelihood function, we obtain the following Fisher Information Matrix (FIM)
\begin{align*}
\matr{J}_L(\phi_t) = \frac{1}{\sigma_r^2}\begin{bmatrix}
	\sum_{i = 1}^L \cos^2\theta_i &   \sum_{i = 1}^L \cos\theta_i \sin\theta_i\\[7pt]
	\sum_{i = 1}^L \cos\theta_i \sin\theta_i &   \sum_{i = 1}^L \sin^2\theta_i
\end{bmatrix},
\end{align*}
where $\cos\theta_i = \frac{x_i - x_t}{d_i}$ and $\sin\theta_i = \frac{y_i - y_t}{d_i}$.

\begin{remark}
Note that if the target is placed at the origin, then the angles here, \emph{i.e.}, the $\theta_i$'s, are a particular realization of the anchor node angles from Definition \ref{ANA}. 
\end{remark}

Taking the inverse of the FIM above, then the CRLB for any unbiased estimator, $\hat{\psi_t} = [\hat{x_t}, \hat{y_t}]^T$, of the target position, $\psi_t = [x_t, y_t]^T$, is given by
\vspace{-3pt}
\begin{align} \label{CRLB_TOA}
\text{CRLB} = \text{tr}\Big(\matr{J}_L^{-1}(\psi_t)\Big),
\end{align}
as defined in Ch. 2.4.1 of \cite{Handbook}.

\begin{definition} \label{DefnOfS}
We define our \emph{localization performance benchmark}, $S$, to be the square root of the CRLB in (\ref{CRLB_TOA}):
\vspace{-8pt}
\begin{align} \label{Sdef}
S \triangleq \sqrt{ \text{tr}\Big(\matr{J}_L^{-1}(\psi_t)\Big) }.
\end{align}
\end{definition}

\begin{remark}
This benchmark is often referred to as the \emph{position error bound (PEB)} in the literature, \cite{PEB1}, \cite{PEB2}.  
\end{remark}

To conclude, notice that from (\ref{Sdef}), a closed-form expression for $S$ can be obtained: 
\vspace{-3pt}
\begin{align} \label{Closed_Form}
S = \sigma_r ~ \frac{ \sqrt{L} } {\sqrt{ \smashoperator \sum_{i = 1}^L \cos^2 \theta_i \sum_{i = 1}^L \sin^2 \theta_i - \Bigg( \sum_{i = 1}^L \cos \theta_i \sin \theta_i \Bigg)^2}} ~,
\end{align}
which is a function of the anchor node angles.


\subsection{Departure From the Traditional Localization Setup}  \label{Departure}


	In the previous section, we assumed a traditional setup, in which the number of participating anchor nodes ($L$) and their positions ($\Psi_L$), as well as the target position ($\psi_t$), were all \emph{fixed}.  In this section however, we now assume that both $\Psi_L$ and $L$ are random and briefly describe how the random placement of anchors impacts our localization performance benchmark, $S$, and how the randomness of $L$ signals a departure from the existing literature. 
	
	We begin by describing how the random placement of anchors affects our localization performance benchmark.  This is accomplished by now invoking Assumption \ref{HPPP} and examining the impact that this has on the expression for $S$ given by (\ref{Closed_Form}). Recall from Section \ref{Impact_of_Assumptions} that Assumption \ref{HPPP} implies that $\psi_t = [0, 0]^T$ and that the anchor node angles are i.i.d. on $[0, 2\pi)$.  \emph{Thus, the $\theta_i$ realizations in (\ref{Closed_Form}) are now replaced with the random variables, $\Theta_i$, from Definition \ref{ANA}.  Since $S$ is now a function of random variables, it itself becomes a random variable, and we may now seek its distribution}. 

	While work in the past has sought this distribution for $S$, \emph{e.g.}, \cite{TOA}, \cite{Fengyu_Zhou}, there always remained one implicit assumption: a \emph{fixed} $L$.  Therefore, the results presented thus far in the literature have only applied to localization setups with a fixed number of anchor nodes, and hence were not applicable network-wide.  To address this issue, we now consider $L$ to be a random variable, whose distribution statistically quantifies the number of anchor nodes participating in a localization procedure. This new interpretation of $L$ will consequently allow us to decouple $S$ from $L$, thereby enabling the marginal distribution of $S$ to account for \emph{all} possible positioning scenarios within a network.  \emph{In addition to the contributions outlined in Section \ref{ContributionsSection}, taking advantage of this new interpretation of L, to subsequently obtain the marginal distribution of S, is the main contribution setting this work apart from the existing literature.}


\subsection{Approximation of the CRLB}  \label{ApproxCRLB}
	
Before deriving the conditional distribution of $S$ given $L$, we use this section to acquire an approximation of our expression for $S$ in (\ref{Closed_Form}).  This will consequently allow for an accurate, tractable, closed-form expression for the conditional distribution of $S$ to be obtained. 
	
\subsubsection{Approximation Preliminaries and Goals}
	
	To facilitate the search for an accurate approximation, we recall that (\ref{Closed_Form}) is now in terms of the random variables $\Theta_i$, and thus can be rewritten using the internodal angles from Definition \ref{Internodal_Angles}.  This given by
\begin{align} \label{S_Internodal_Angles}
S = \sigma_r ~ \frac{\sqrt{L}}{\sqrt{ \smashoperator \sum_{i = 1}^{L-1} \sum_{j = i + 1}^L \sin^2 \Bigg( \sum_{k = i}^{j-1} \angle_k \Bigg)}}~,
\end{align}

\begin{definition} \label{DefD}
We define the terms underneath the square-root in the denominator in (\ref{S_Internodal_Angles}) by the random variable $D$.  That is,
\vspace{-7pt}
\begin{align} \label{D_Expression}
D = \sum_{i = 1}^{L-1} \sum_{j = i + 1}^L \sin^2 \Bigg( \sum_{k = i}^{j-1} \angle_k \Bigg).
\end{align}
\end{definition}
	
	Thus, we would like to find an approximation for $D$ which comprises two key traits: (1) it allows for a straightforward transformation of random variables, \emph{i.e.}, the number of $\sin^2(\cdot)$ terms does not change with $L$ (and would ideally only involve a single term), and (2) it simultaneously does not sacrifice accuracy, \emph{i.e.}, the approximation should preserve as much ``information'' as possible about the setup of anchors, implying that the approximation should dominate (or contribute the most to) the total value of $D$. 
	
\subsubsection{Initial Approach and Intuition}

	In trying to find an approximation that satisfies both traits, we consider the following possibilities.  First, consider approximating $D$ with the sine squared of an arbitrary internodal angle, $\sin^2(\angle_k)$, or with a sum of consecutive internodal angles, $\sin^2(\angle_k + \angle_{k+1} + ...)$, where the starting angle, $\angle_k$, is arbitrary. While these possible approximations may seem like reasonable candidates for satisfying the first trait, they unfortunately fall short of satisfying the second trait.  To see why, it is illustrative to examine Fig. \ref{Setup2} under different realizations/placements of anchor nodes.  By only looking at the same unordered internodal angle on every realization, little knowledge is gained about the total setup of anchors.  
For example, on one realization, the arbitrary internodal angle being examined might be large and would therefore give a strong indication of how the rest of the anchors are placed, however, on another realization, this same internodal angle might be small, thus giving little information about the placement of the remaining anchors.  Hence, in general, arbitrary internodal angles
do not provide accurate approximations due to their inconsistency in describing the anchor node setup, which consequently leads to their (sine squared terms') inability to capture the total value of $D$ across all realizations of anchors for a given $L$.

\subsubsection{A Quantitative Approach Using Mutual Information}

	Taking advantage of the intuition gained above, it should now be clear that we would like an approximation that utilizes angles which tend to \emph{consistently} dominate any given setup.  Therefore, it makes sense to examine larger internodal angles, and thus, the use of internodal angle \emph{order statistics} follows naturally.  Since we ideally desire a single-term approximation, we examine the possibility of using the sine-squared of the largest, second largest, and third largest internodal angles, \emph{i.e.},\footnote{We will not examine \emph{sums} of the larger internodal angle order statistics, \emph{e.g.}, $\sin^2(\angle_{(L)} + \angle_{(L-1)})$ or $\sin^2(\angle_{(L)}) +\sin^2(\angle_{(L-1)})$, for the following two reasons: (1) they would lead to a more complex expression for the conditional distribution of S than desired, and (2) they offer no gain in accuracy over using the sine squared of just a \emph{single} internodal angle order statistic, as evidenced through simulation.} 
\begin{align} \label{Ws}
W_{(L)} = \sin^2&(\angle_{(L)}), ~~~~ W_{(L-1)} = \sin^2(\angle_{(L-1)}), ~~~~\text{and} ~~~~ W_{(L-2)} = \sin^2(\angle_{(L-2)}).
\end{align}

	Note that these larger internodal angles should intuitively contain more ``information'' about the setup of anchors, and consequently $D$, since they greatly restrict the placement of the remaining anchors.  Since each of the order statistic approximations above might seem viable under this qualitative notion of information, we thus turn towards a quantitative notion, in order to justify the use of one of these approximations. 
	
	Towards this end, we utilize the concept of \emph{mutual information}.  The reason behind this choice, over that of correlation for example, is because mutual information captures \emph{both} linear and nonlinear dependencies between two random variables, since it is zero if and only if the two random variables are independent \cite{Jeannette}.  Hence, we examine the mutual information between $D$ and the random variables $W_{(L)}$, $W_{(L-1)}$, and $W_{(L-2)}$, so that we may quantify which approximation carries the most information about $D$.  Thus, we \emph{condition} on $L$ equaling some integer $\ell (\geq 3)$ and calculate:
\vspace{-10pt}
\begin{align}
I(D; W_{(i)} | L=\ell) = h(D|L=\ell) - h(D | W_{(i)}, L=\ell),
\end{align}
where $i \in \{L, L-1, L-2\}$, and the differential entropies are given by
\begin{align}
h(D|L=\ell) &= \int\limits_\mathcal{D} \! \! f_D(d|\ell) \log_2 \frac{1}{f_D(d|\ell)} ~ \text{d} d, \label{EntropyOfD} ~~~~~\text{and} \\
h(D | W_{(i)}, L=\ell) &= \! \! \!\int\limits_{\mathcal{W}_{(i)}}\! \!\int\limits_\mathcal{D} \! \! \! f_{DW_{(i)}}(d,w|\ell) \log_2 \frac{1}{f_{D}(d|w,\ell)} ~ \text{d} d \,\text{d} w, \label{EntropyOfDGivenW}
\end{align}
where $\mathcal{W}_{(i)}$ and $\mathcal{D}$ are the supports of $W_{(i)}$ and $D$, respectively \cite{CoverAndThomas}.  The mutual information between the approximations in (\ref{Ws}) and $D$ are given in Fig. \ref{MutualInformations}, versus $L$.
\begin{figure}[t]
\centering
\includegraphics [scale=0.60]{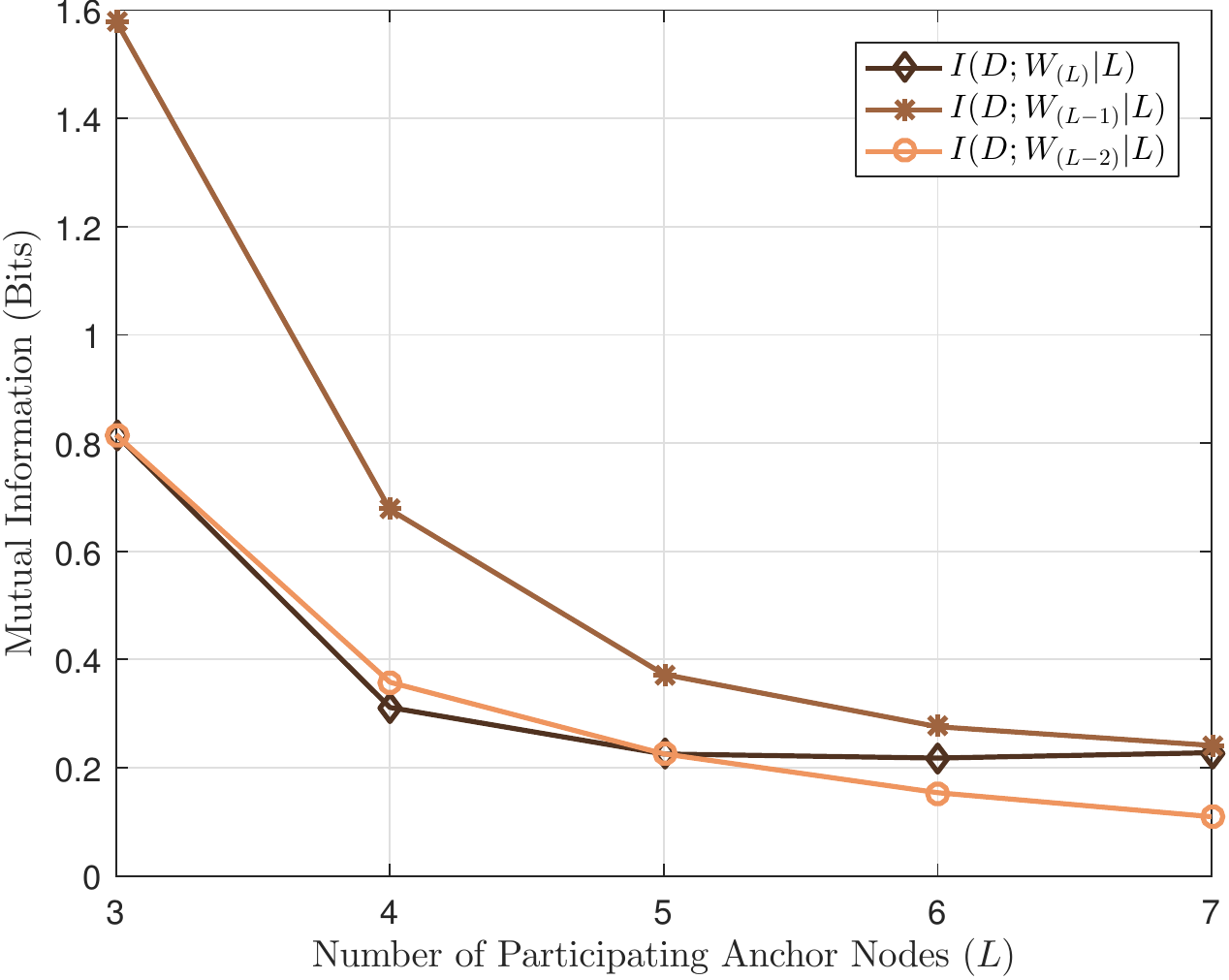}
\vspace{-10pt}
\caption{\textsc{Justifying Approximations Through Mutual Information}. The mutual informations were calculated numerically by computing (\ref{EntropyOfD}) and (\ref{EntropyOfDGivenW}), where the necessary distributions were generated using a Monte Carlo simulation of 10 million anchor node realizations. The bin width of these distributions was chosen to be 0.01, Matlab's `spline' option was used to interpolate the integrands in (\ref{EntropyOfD}) and (\ref{EntropyOfDGivenW}), and the supports of $D$ and $W_{(i)}$ are given by $\mathcal{D} = [0, L^2/4]$ and $\mathcal{W}_{(i)} = [0,1]$, where $i \in \{L, L-1, L-2\}$. Furthermore, we adopt the convention: $0 \log_2 0 = 0$ ``based on continuity arguments'' \cite{CoverAndThomas}.}
\label{MutualInformations}
\end{figure}
From Fig. \ref{MutualInformations}, it is evident that the mutual information between $D$ and the approximation $W_{(L-1)}$ is the highest across all numbers of participating anchors shown.  

\subsubsection{Investigating the High Mutual Information of D and $W_{(L-1)}$}

	To explore the reasoning behind this result, we examine the effect that $W_{(L-1)}$ has on the total value of $D$.  We begin by rewriting $D$ as follows: 
	
\begin{proposition} \label{RewritingD}
The random variable D from Definition \ref{DefD} can be equivalently expressed as
\begin{align*}
D = \sum_{k=1}^L \sin^2 \big(\angle_{(k)} \big) + \sum_{i = 2}^{L-2} \sum_{j = 1}^{L-i} \sin^2 \Bigg( \sum_{k = j}^{j+i-1} \angle_k \Bigg).
\end{align*}
\end{proposition}

\begin{proof}
See Appendix \ref{PropositionProof}.
\end{proof}

	
	By separating the sine-squared terms of the internodal angle order statistics from the total sum, Proposition \ref{RewritingD} makes it clearer as to how our approximations from (\ref{Ws}) may affect the total value of $D$.  To reveal the effects of $W_{(L-1)}$ in particular, we present the following lemma along with its corollaries: 
	
\begin{lemma}  \label{Angle_Max2}
The cdf of the second largest order statistic of the internodal angles, $\angle_{(L-1)}$, conditioned on $L$, is given by
\begin{align}
F_{\angle_{(L-1)}} \big(\varphi \mid L \big) = \sum_{n=0}^\mathcal{X} (-1)^{n-1} {L \choose n} (n-1) \Bigg(1-\frac{n\varphi}{2\pi}\Bigg)^{L-1},
\end{align}
where $\mathcal{X} = \mathrm{min}\big\{L, \floor*{2\pi / \varphi}\big\}$ and the support is $0 \leq \angle_{(L-1)} \leq \pi$. (Note $L\geq2$, since if $L<2$, $\angle_{(L-1)}$ would not exist.)
\end{lemma}

\begin{proof}
We refer the reader to the conference version of this paper, \emph{i.e.}, Appendix B of \cite{ICC_Paper}.
\end{proof}

\begin{corollary} \label{ExpectationOfSecondLargest}
Given a finite L, the expected value of the second largest order statistic of the internodal angles, $\angle_{(L-1)}$, conditioned on $L$, is given by
\begin{align}
\E[\angle_{(L-1)} | L] = \sum_{n=2}^L (-1)^n {L \choose n} \Bigg(\frac{2\pi(n-1)}{nL}\Bigg)
\end{align}
\end{corollary}

\begin{proof}
See Appendix \ref{CorProof}.  
\end{proof}

\begin{corollary} \label{VarianceOfSecondLargest}
Given a finite L, the variance of $\angle_{(L-1)}$, conditioned on $L$, is given by
\begin{align}
\text{\emph{VAR}}[\angle_{(L-1)} | L] = \frac{4\pi^2}{L} \sum_{n=2}^L (-1)^n {L \choose n} \Bigg(\frac{n-1}{n}\Bigg) \Bigg[\frac{2}{n(L-1)} + \frac{c}{L}\Bigg],
\end{align}
where $c = \sum_{m=2}^L (-1)^{m+1} {L \choose m} \Big(\frac{m-1}{m}\Big)$.
\end{corollary}

\begin{proof}
Note: $\text{VAR}[\angle_{(L-1)} | L] = \E[\angle_{(L-1)}^2 | L] - \Big(\E[\angle_{(L-1)} | L]\Big)^2$,
where the derivation of $\E[\angle_{(L-1)}^2 | L]$ is analogous to that of $\E[\angle_{(L-1)} | L]$ in the proof of Corollary \ref{ExpectationOfSecondLargest}.
\end{proof}
%

	Next, we plot Corollary \ref{ExpectationOfSecondLargest}, plus/minus one and two standard deviations of $\angle_{(L-1)}$. This is given in Fig. \ref{ExpectationOfAngleMax2}, versus $L$.  Here, we can see that the second largest internodal angle is centered and concentrated around $\pi/2$, suggesting that $W_{(L-1)} = \sin^2(\angle_{(L-1)})$ will be concentrated about its maximum of one.  This implies that, for the majority of anchor node placements, $\sin^2(\angle_{(L-1)})$ will be a dominant term in the expression for $D$ in Prop. \ref{RewritingD}.  Thus, the $\sin^2(\angle_{(L-1)})$ term will tend to contribute the most, that a given $\sin^2(\cdot)$ term could contribute, to the total value of $D$.  This is especially true for small values of $L$, which is our focus.

	Also for low $L$, $\angle_{(L-1)}$ is intuitively the dominant angle in that it places the greatest constraints on the remaining angles.  That is, once $\angle_{(L-1)}$ is determined, it restricts both $\angle_{(L)}$ and $\angle_{(L-2)}$, and thus gives the greatest sense of the total setup of anchors.  Note, when considering order statistics other than $\angle_{(L-1)}$, the constraints placed on the remaining angles are not as pronounced.  Furthermore, by examining different realizations of anchors (Fig. \ref{Setup2} as an example), along with Prop. \ref{RewritingD}, one can see that when $W_{(L-1)}$ is small ($\angle_{(L-1)} \approx 0 \text{ or } \pi$), than so is $D$, and when $W_{(L-1)}$ is large ($\angle_{(L-1)} \approx \pi/2$), a large value of $D$ follows.  Thus, $W_{(L-1)}$'s consistency as a dominant term in Prop. \ref{RewritingD}, along with its intuitive correlation with $D$, offer supporting evidence as to why $I(D;W_{(L-1)}|L=\ell)$ is higher than both $I(D;W_{(L)}|L=\ell)$ and $I(D;W_{(L-2)}|L=\ell)$ for low $L$.

	In summary, mutual information has proved its utility by revealing that $W_{(L-1)}$ is perhaps the best approximation of $D$, for the desirable lower values of $L$.  Since $W_{(L-1)}$ possesses the two desirable traits for an approximation, discussed at the beginning of this section, we henceforth use $W_{(L-1)}$ in our approximation of $D$.

\subsubsection{Completing the Approximation}
	
	To complete the approximation of $D$, and consequently $S$, all that now remains is to ensure that $D$ and $W_{(L-1)}$ have the same range of possible values, \emph{i.e.}, the same support.  This will ensure that our ultimate approximation for $S$ will produce the same range of values as the true $S$.  In order to accomplish this, we approximate $D$ with a scaled version of $W_{(L-1)}$, \emph{i.e.}, $D \approx k\, W_{(L-1)}$,
and thus search for the value of the constant $k$ so that $kW_{(L-1)}$ yields the desired support. Since $\mathcal{D} = [0, d_\text{max}]$ and $\mathcal{W}_{(L-1)} = [0, 1]$ (which follows from the support of $\angle_{(L-1)}$, Lemma \ref{Angle_Max2}), then in order to have the support of $kW_{(L-1)}$ equal that of $\mathcal{D}$, we simply need to set $k = d_\text{max}$.  The value of $d_\text{max}$ is presented in the following lemma:
\begin{lemma} \label{Dmax}
Let L be finite.  If D is given as in (\ref{D_Expression}), then its maximum value is $d_{\text{max}} = L^2/4$.
\end{lemma}
\vspace{-5pt}
\begin{proof}
See Appendix \ref{DmaxProof}.
\end{proof}
\noindent Thus, we now have $D \approx (L^2/4) \cdot W_{(L-1)}$, which completes our approximation of $D$.
Lastly, substituting this approximation for $D$ into the expression for $S$ in (\ref{S_Internodal_Angles}) finally yields:

\begin{approximation} \label{approximation1}
The localization performance benchmark, $S$, can be approximated by
\vspace{-5pt}
\begin{align}
S \approx \sigma_r \cdot \sqrt{\frac{4}{L}} \cdot \frac{1}{\sin (\angle_{(L-1)})} ~ ,
\end{align}
where $0 \leq \angle_{(L-1)} \leq \pi$, as stated in Lemma \ref{Angle_Max2}.
\end{approximation}


\vspace{-5pt}
\subsection{The Conditional CRLB Distribution} \label{Cond_CRLB_Dist}

	
\begin{figure}[t]
\centering
\begin{minipage}[t]{0.48\textwidth}
\centering
\includegraphics [scale=0.60]{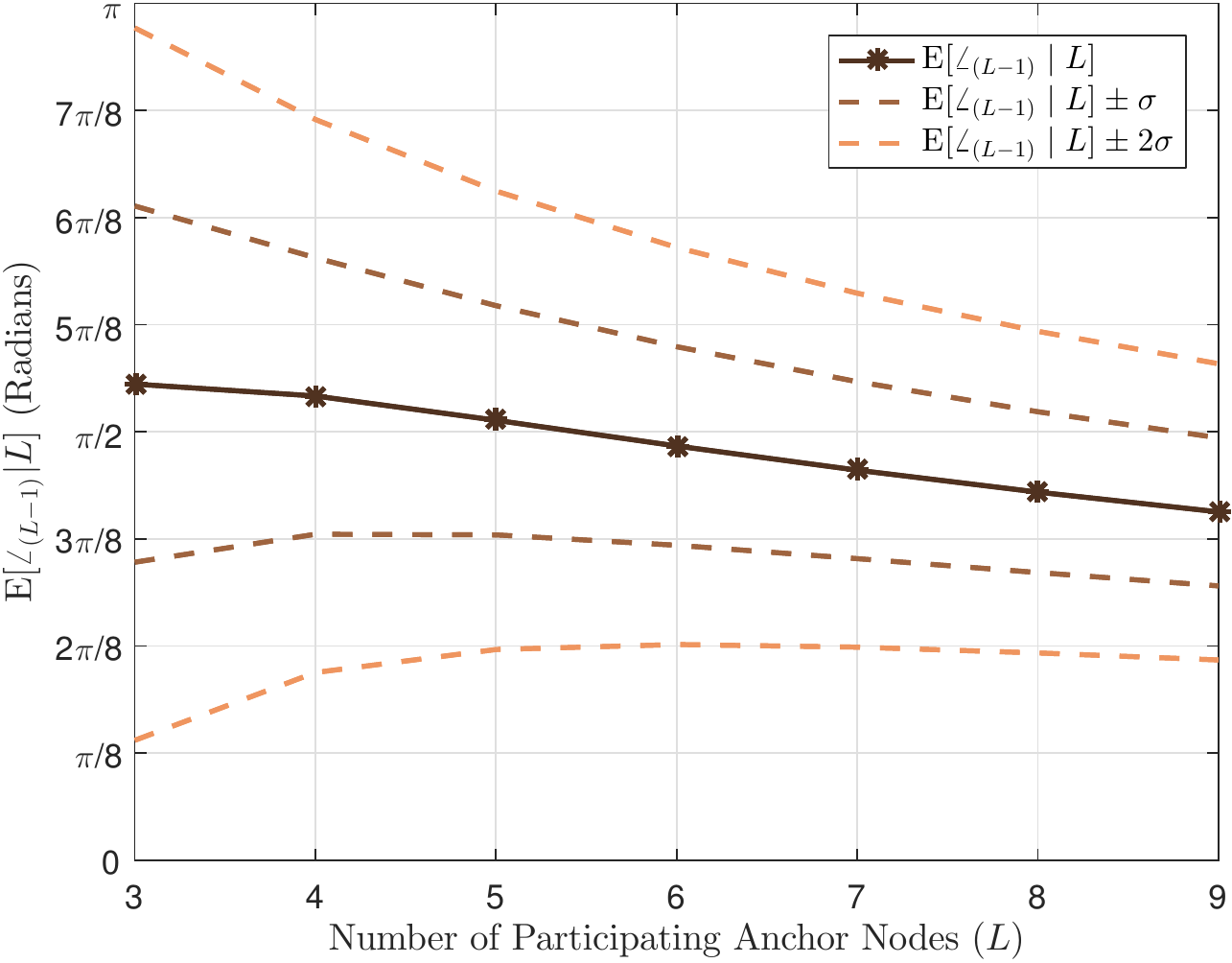}
\caption{\textsc{The Second Largest Internodal Angle Order Statistic.} This figure gives a sense of the concentration of the distribution of $\angle_{(L-1)}$ around $\pi/2$.  $\E[\angle_{(L-1)} | L]$ is given by Cor. \ref{ExpectationOfSecondLargest} and $\sigma = \sqrt{\text{VAR}[\angle_{(L-1)} | L]}$ is from Cor. \ref{VarianceOfSecondLargest}.}
\label{ExpectationOfAngleMax2}
\end{minipage}\hfill
\begin{minipage}[t]{0.48\textwidth}
\centering
\includegraphics [scale=0.60]{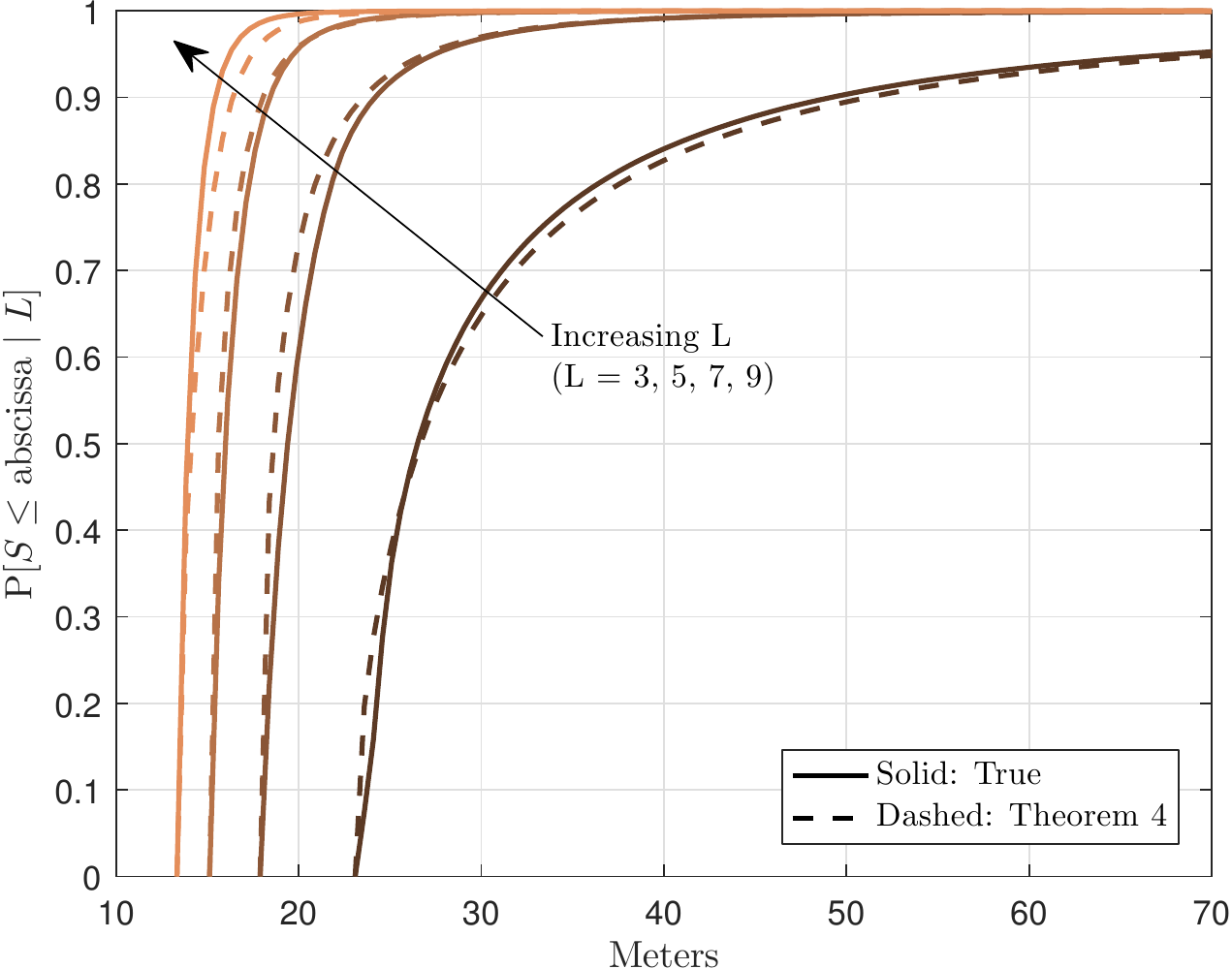}
\caption{\textsc{Accuracy of Theorem \ref{SgivenL}}. The true conditional cdf of $S$ given $L$ was generated using a Monte Carlo simulation of (\ref{S_Internodal_Angles}) over 1 million random setup realizations of the internodal angles. Note, $\sigma_r = 20 \, m$.}
\label{sqrtCRLB_CDFs}
\end{minipage}
\end{figure}	
		
\begin{theorem} \label{SgivenL}
If the localization performance benchmark, $S$, is given by Approximation \ref{approximation1}, then the cdf of $S$ conditioned on $L$ is 
\begin{align} \label{SgivenL_cdf}
F_S(s \mid L, \sigma_r) &= \sum_{n=0}^{\mathcal{X}_2} (-1)^{n-1} {L \choose n} (n-1) \Bigg(1-\frac{n \varphi_2 }{2\pi}\Bigg)^{L-1} - \sum_{n=0}^{\mathcal{X}_1} (-1)^{n-1} {L \choose n} (n-1) \Bigg(1-\frac{n \varphi_1 }{2\pi}\Bigg)^{L-1}, 
\end{align}
where $\varphi_1 = \sin^{-1} (a/s)$, $\varphi_2 = \pi - \sin^{-1} (a/s)$, $a = \sigma_r \cdot \sqrt{4/L}$, $\mathcal{X}_1 = \mathrm{min}\big\{L, \floor*{2\pi/\varphi_1}\big\}$, $\mathcal{X}_2 = \mathrm{min}\big\{L, \floor*{2\pi/\varphi_2}\big\}$, 
and the support is $S \in [a, \infty)$.
\end{theorem}

\begin{proof}
See Appendix \ref{TheoremProof}. 
\end{proof}

\begin{remark}
	Although Theorem \ref{SgivenL} is the conditional distribution of our approximation of $S$, it provides two clear advantages over the true conditional distribution presented in \cite{Fengyu_Zhou} and over the approximate conditional distribution presented in \cite{TOA}.  First, Theorem \ref{SgivenL} offers a simple, closed-form, algebraic expression involving only finite sums, as opposed to the rather complex expression in \cite{Fengyu_Zhou} involving an improper integral of products of scaled Bessel functions.  Second, Theorem \ref{SgivenL} is remarkably accurate for lower numbers of participating anchor nodes, see Fig. \ref{sqrtCRLB_CDFs}.  This comes in contrast to the approximate distribution presented in \cite{TOA}, which was derived asymptotically and therefore only accurate for higher numbers of participating anchors.  This selective accuracy of Theorem \ref{SgivenL} is desirable since a device is more likely to hear lower numbers of participating anchors, especially in infrastructure-based, terrestrial wireless networks.\footnote{By infrastructure-based wireless networks, we mean any wireless network setup with mobile devices, fixed access points, and separate uplink/downlink channels.} 
\end{remark}


\subsection{The Distribution of the Number of Participating Anchors}  \label{JaviersWork}

	The next step needed to achieve our goal is to find the distribution of the number of participating anchors, $f_L(\ell)$.  In this section, we build upon the localizability results from \cite{Javier_Journal} in order to obtain this distribution.  Towards this end, we present the relevant theorems from this work and modify them for our use here.  Finally, we conclude with a discussion on the applicability of these results.

\subsubsection{Overview of Localizability Work}
	
	Recall from Section \ref{RelatedWork} that the goal of \cite{Javier_Journal} was to derive an expression for the probability that a mobile can hear at least $\ell$ base stations for participation in a localization procedure in a cellular network, \emph{i.e.}, $\text{P}[L \geq \ell]$.  To derive this expression, the authors assumed that base stations were placed according to a homogeneous PPP, and then examined the SIRs of the base station signals received at a ``typical user'' placed at the origin.\footnote{SIR = Signal to Interference Ratio.  Noise is ignored since \cite{Javier_Journal} assumes an interference-limited network.}  Specifically, they examined the SIR of the $\ell^\text{th}$ base station (denoted $\text{SIR}_\ell$), since this was used directly to determine $\text{P}[L \geq \ell]$.
	
	Since $\text{SIR}_\ell$ depends on the locations of interfering base stations, then the base stations' placement according to a PPP implies that $\text{SIR}_\ell$ becomes a random variable.  Consequently, its distribution also becomes a function of the PPP density, $\lambda$.  Additionally, the authors incorporate a network loading parameter, $q$, into $\text{SIR}_\ell$, where $0 \leq q \leq 1$.  This means that any given base station can be considered active (\emph{i.e.}, interfering with base station $\ell$'s signal) with probability $q$.  Furthermore, $\text{SIR}_\ell$ is also a function of pathloss, $\alpha$, where $\alpha > 2$, and the distances of base stations to the target.
	
	With $\text{SIR}_\ell$ statistically characterized, the authors were able to determine $\text{P}[L \geq \ell]$ by noting that $\text{P}[L \geq \ell] = \text{P}[\text{SIR}_\ell \geq \beta/\gamma]$, where $\beta/\gamma$ is the pre-processing SIR threshold for detection of a signal.\footnote{Note that for $q<1$, this equality holds for all but a few rare corner cases.  However, the probability of these cases occurring is vanishingly small and thus has little to no impact on the accuracy of the subsequent localizability results \cite{Javier_Journal}.}  Here, $\gamma$ is the processing gain at the mobile (assumed to also average out the effect of small scale fading), and $\beta$ is the post-processing threshold.  Thus, since $\text{P}[L \geq \ell]$ = $\text{P}[\text{SIR}_\ell \geq \beta/\gamma]$, then $\text{P}[L \geq \ell]$ must also depend on all of the network parameters described above. We denote this dependency by $\text{P}[L \geq \ell \mid \alpha, \lambda, q, \gamma, \beta]$.
	
	
	Before continuing to the localizability results, we mention one last caveat regarding the PPP density.  That is, when shadowing is present, it can easily be incorporated into the PPP network model through small displacements of the base station locations.  This results in a new PPP density, which accounts for this effect of shadowing.  This new shadowing-transformed density is given by $\tilde{\lambda} = \lambda \, \E[\mathcal{S}_z^{2/\alpha}]$, where $\mathcal{S}_z$ is assumed to be a log-normal random variable representing the effect of shadowing on the signal from base station $z$ to the origin \cite{Shadowing}.\footnote{The $\mathcal{S}_z$ are assumed to be i.i.d. $\forall z$. $\mathcal{S}_z$'s log-normal behavior implies it is distributed normally when expressed in dB \cite{Shadowing}.}  Thus, by using $\tilde{\lambda}$, we incorporate shadowing under the log-normal model presented in Section II of \cite{Shadowing}.

\subsubsection{The Localizability Results}
	
	In this section, we present the main theorem which will enable us to obtain $f_L(\ell)$. 
\vspace{-7pt}
\begin{lemma} (Theorem 2, \cite{Javier_Journal}) \label{JavierThrm2}
The probability that a mobile device can hear at least $\ell$ base stations for participation in a localization procedure is given by
\vspace{-8pt}
\begin{flalign*}
\hspace*{-8.23cm} \text{\emph{P}}[L \geq \ell \mid \alpha, \tilde{\lambda}, q, \gamma, \beta] =
\end{flalign*}
\begin{flalign*}
\left( \! 1 \! -  \! \! \sum_{n=0}^{\ell-1} \! e^{-\frac{\alpha - 2}{2 q \beta/\gamma}} \frac{ \Big(\frac{ \alpha -2 }{2q\beta/\gamma} \Big)^n} {n!} \right) \! f_\Omega(0) + \frac{4(\tilde{\lambda}\pi)^\ell}{(\ell-1)!} \sum_{\omega=1}^{\ell-1} \! f_\Omega(\omega) \! \! \int_0^\infty \! \! \! \! & \int_0^{r_\ell} \! \! \! \mathbf{1} \! \left[ \frac{r_\ell^{-\alpha}}{r_1^{-\alpha} \! +\! \frac{2(\omega-1)}{2-\alpha} \! \cdot \! \frac{r_\ell^{2-\alpha} - r_1^{2-\alpha}}{r_\ell^2 - r_1^2} \! + \! \frac{2\pi q \tilde{\lambda}}{\alpha - 2} r_\ell^{2-\alpha} } \! \geq \! \frac{\beta}{\gamma} \right] \nonumber \\[4pt]
& \times r_1(r_\ell^2-r_1^2)^{\omega-1} r_\ell^{2(\ell-\omega)-1} \omega e^{-\tilde{\lambda} \pi r_\ell^2} ~ \text{\emph{d}} r_1 ~ \text{\emph{d}} r_\ell,
\end{flalign*}
\vspace{-5pt}
where $\ell \in \{1, 2, \dots \}$, and $\Omega$ is a random variable denoting the number of active participating base stations interfering with the $\ell^\text{\emph{th}}$. Note, $\Omega \sim \text{Binomial}(\ell-1, q) = f_\Omega(\omega)$.  Additionally, for the trivial case of $\ell = 0$, we define $\text{\emph{P}}[L \geq 0] = 1$. (Note: $\gamma$ and $\beta$ are in linear terms, not dB.)
\end{lemma}

\begin{remark}
This theorem was derived under the following assumptions: (1) a dominant interferer and (2) interference-limited networks.  We refer the reader to Section III-D of \cite{Javier_Journal} for further details regarding these assumptions and the consequent derivation of this theorem.
\end{remark}

\subsubsection{The Distribution of L}	

	With these localizability results, we now finally present the distribution of the number of participating anchors.
\vspace{-3pt} 
\begin{theorem} \label{DistOfL}
The pdf of L is given by
\vspace{-5pt}
\begin{align*}
f_L(\ell \mid \alpha, \tilde{\lambda}, q, \gamma, \beta) = \text{\emph{P}}[&L \geq \ell \mid \alpha, \tilde{\lambda}, q, \gamma, \beta] - \text{\emph{P}}[L \geq \ell +1 \mid \alpha, \tilde{\lambda}, q, \gamma, \beta],
\end{align*}
where the support is $L \in \{0, 1, 2, \dots\}$ and the probabilities are given by Lemma \ref{JavierThrm2}.
\end{theorem}

\vspace{-5pt}
\subsubsection{The Distribution of L with Frequency Reuse}

	Using Theorem \ref{DistOfL}, we may obtain another expression for $f_L(\ell)$ which incorporates a frequency reuse parameter, $\mathcal{K}$.  This parameter models the ability of base stations to transmit on $\mathcal{K}$ separate frequency bands, thereby limiting interference to a per-band basis.  This can easily be incorporated into the model by considering $\mathcal{K}$ independent PPPs whose densities are that of the original PPP divided by $\mathcal{K}$. Thus, if $n_k$ is the number of participating base stations in band $k$, then the total number of participating base stations is given by $L = \sum_{k=1}^\mathcal{K} n_k$.  Thus, to find $\text{P}[L=\ell]$ under frequency reuse, we simply need to account for all of the per-band combinations of participating base stations such that their sum equals $\ell$.  This is given in the following corollary, which is a modification of Theorem 3 of \cite{Javier_Journal}. 
	
\begin{corollary} \label{CorFreqReuse}
The pdf of L, given a frequency reuse factor of $\mathcal{K}$, is 
\vspace{-3pt}
\begin{align*}
f_L(\ell \mid \alpha, \tilde{\lambda}, q, \gamma, \beta, \mathcal{K}) = \sum_{\substack{\{n_1, \dots, n_\mathcal{K}\} \\ \sum_i n_i = \ell}} \prod_{k=1}^\mathcal{K} f_L \Bigg(\!~n_k \! ~\Bigg|~\! \alpha, \frac{\tilde{\lambda}}{\mathcal{K}}, q, \gamma, \beta \Bigg) ,
\end{align*}
where the multiplicands are given by Thm. \ref{DistOfL}, $\mathcal{K} \in \{1, 2, \dots\}$, and the support is $L \in \{0, 1, 2, \dots\}$.
\end{corollary}

\begin{remark}
When $\mathcal{K}=1$, this corollary reduces to Theorem \ref{DistOfL}.  Further, this corollary may be evaluated numerically through the use of a recursive function.
\end{remark}

\subsubsection{Applicability of the Results}

	Now that we have obtained the pdf of $L$, we conclude with a brief discussion regarding its applicability.  We begin by taking note of the support of $f_L(\ell)$.  Whereas up until this section we have proceeded under the assumption that the target is localizable, the support of $f_L(\ell)$ now allows us to consider cases where the target is unlocalizable, \emph{i.e.}, $L = 0, 1, 2$.  Thus, as will be addressed in the following section, we may use these cases to determine the percentage of the network where a target is unlocalizable.  

	 Lastly, we note that while the localizability results in \cite{Javier_Journal} (Lemma \ref{JavierThrm2}) were presented in the context of cellular networks, these results are actually applicable to any infrastructure-based wireless network using downlink measurements, so long as the distribution parameters are altered accordingly.  This implies that the distribution for $L$ (Corollary \ref{CorFreqReuse}) also has this applicability, since it was derived using Lemma \ref{JavierThrm2}.  \emph{Thus, since we use Corollary \ref{CorFreqReuse}, along with a modified Theorem \ref{SgivenL}, to derive the marginal distribution of $S$, then this final distribution will also be applicable to any infrastructure-based wireless network employing a TOA localization strategy.}


\subsection{The Marginal CRLB Distribution}  \label{MarginalDistribution}


	In this section, we modify Theorem \ref{SgivenL} and combine this with Corollary \ref{CorFreqReuse} to obtain the marginal distribution of $S$.  First, we state one last network assumption often used in practice:
	
\begin{assumption} \label{Finite_Num_Participating_ANs}
For a given localization procedure, only a finite number of anchor nodes, $N$ ($\geq 3$), are ever tasked to transmit localization signals.
\end{assumption}

\begin{remark}
Just because $N$ anchors are tasked, does not mean that all $N$ signals are necessarily heard.
\end{remark}

Under this assumption, and further considering scenarios where the target is unlocalizable, we modify Theorem \ref{SgivenL} as follows:
\vspace{-8pt}
\begin{align} \label{ModifiedSgivenL}
F'_S(s \mid L, \sigma_r) = 
\begin{cases}
F_S(s \mid L = N, \sigma_r) & L  \geq N \\[-7pt]
F_S(s \mid L, \sigma_r) & L = 3, \dots, N-1\\[-7pt]
u(s-M)          & L = 0, 1, 2
\end{cases}
\end{align}	
where $F'_S$ is the new modified conditional distribution of $S$, $F_S$ is the previous conditional distribution given in Theorem \ref{SgivenL}, and $M \in \mathbb{R}^+$ is a predetermined localization error value used to account for unlocalizable scenarios (described in more detail below).

\begin{remark}
While Theorem \ref{SgivenL} only accounted for scenarios in which the target was localizable, this modified form, however, now accounts for unlocalizable scenarios.  For these scenarios, this modified conditional distribution yields a step function, which is a valid cdf and corresponds to a deterministic value for the localization error, \emph{i.e.}, $\text{P}[S = M \mid L] = 1$.  Thus, we account for cases where the target is unlocalizable by assigning an arbitrary localization error value for $M$, which is chosen to represent cases where there is ambiguity in a target's position estimate.\footnote{In this paper, we choose a value for $M$ that applies for all $L<3$. Thus, a large enough $M$ allows for a clear distinction between the localizable and unlocalizable portions of the network through a quick examination of the cdf of $S$.  Note however, that one could account for the $L = 0, 1, 2$ scenarios separately. For example, if $L = 1$ then in a cellular network one may want to choose $M$ to be the cell radius, since the user equipment typically knows in which cell it is located \cite{CellRadius}.}
\end{remark}

\begin{remark}
It is possible that a mobile may hear more anchors than are tasked to perform the localization procedure, \emph{i.e.}, $L>N$.  In this case, the $N$ participating anchors are likely those with the highest received SIR at the target, if connectivity information is known a priori. Therefore, in this scenario, localization performance will only be based on the $N$ anchors tasked.  This is clearly reflected in the modified conditional distribution of $S$ in (\ref{ModifiedSgivenL}).
\end{remark}


	Using this modified Theorem \ref{SgivenL} in (\ref{ModifiedSgivenL}), along with Corollary \ref{CorFreqReuse}, we may now obtain the distribution of localization error for an entire wireless network:
\vspace{-5pt}
\begin{theorem} \label{MainThrm}
The marginal cdf of the localization performance benchmark, S, is 
\vspace{-5pt}
\begin{align*}
F_S(s \mid \sigma_r, \alpha, \tilde{\lambda}, q, \gamma, \beta, \mathcal{K}, M, N) =\, &F_S(s \mid \ell = N, \sigma_r) \, \Big[1 - \text{\emph{P}}[L \leq N-1| \alpha, \tilde{\lambda}, q, \gamma, \beta, \mathcal{K}] \Big] \\[-4pt]
&+ \sum_{\ell=3}^{N-1} F_S(s \mid \ell, \sigma_r) \, f_L(\ell \mid \alpha, \tilde{\lambda}, q, \gamma, \beta, \mathcal{K}) \\[-4pt]
&+ u(s-M) \, \text{\emph{P}}[L \leq 2| \alpha, \tilde{\lambda}, q, \gamma, \beta, \mathcal{K}],
\end{align*}
where $F_S(s \! \mid \! L\! =\! \ell, \, \sigma_r)$ is given by Theorem \ref{SgivenL}, $f_L(\ell  \!\mid \! \alpha, \tilde{\lambda}, q, \gamma, \beta, \mathcal{K})$ is given by Corollary \ref{CorFreqReuse}, and $\text{\emph{P}}[L \leq x | \alpha, \tilde{\lambda}, q, \gamma, \beta, \mathcal{K}] = \sum_{\ell = 0}^{x} f_L(\ell  \!\mid \! \alpha, \tilde{\lambda}, q, \gamma, \beta, \mathcal{K})$.
\end{theorem}

\begin{proof}
Multiplying the modified conditional distribution of $S$, given in (\ref{ModifiedSgivenL}), by the marginal distribution of $L$ from Corollary \ref{CorFreqReuse} gives the joint distribution of $S$ and $L$.  Then, setting $L$ equal to a particular realization, $\ell$, and summing over all realizations, gives us the marginal cdf of $S$, as desired.
\end{proof}


\vspace{-5pt}
\begin{remark}
First, recall that the conditional distribution of our approximation of $S$, $F_S(s \mid L=\ell, \sigma_r)$ given by Theorem \ref{SgivenL}, is accurate for lower $\ell$ values.  Next, note that: (1) $f_L(\ell)$ declines rapidly as $\ell$ increases (an intuitive result, since the probability of hearing many anchor nodes should be small in infrastructure-based wireless networks), and (2) only a maximum of $N$ nodes are tasked to perform a localization procedure (Assumption \ref{Finite_Num_Participating_ANs}). Thus, these two facts validate the use of our approximation, since the cases where our approximation is less than ideal (\emph{i.e.}, for large $L$) will now either be multiplied by $f_L(\ell) \approx 0$, or considered invalid in a realistic network under Assumption \ref{Finite_Num_Participating_ANs}.  As a consequence, Theorem \ref{MainThrm} will also retain this accuracy.  
\end{remark}


\begin{remark}
\emph{We conclude by noting that this distribution accounts for localization error over all setups of anchor nodes, numbers of participating anchors, and placements of a target anywhere in the network.  Hence, this distribution completely characterizes localization performance throughout an entire wireless network and represents the main contribution of this work.}
\end{remark}


\section{Numerical Analysis} \label{NumericalAnalysis}

	In this section, we examine the accuracy of Theorem 7 and investigate how changing network parameters affects localization performance throughout the network.

\vspace{-5pt}	
\subsection{Description of Simulation Setup}
	Here, we discuss the parameters that are fixed across all simulations and include a description of how the simulations were conducted.
	
\subsubsection{Fixed Parameter Choices and their Effect on Model Assumptions}
	
	For these simulations, we consider the case of a cellular network, and thus, we place anchor nodes such that the PPP density matches that of a ubiquitous hexagonal grid with 500m intersite distances (\emph{i.e.}, $\lambda = 2/\big(\sqrt{3} \cdot 500^2 \, m^2 \big)$) \cite{Javier_Journal}. Furthermore, we choose a shadowing standard deviation of 8 dB, which defines our shadowing-transformed density parameter, $\tilde{\lambda}$.
		
	Next, we set our pathloss exponent, $\alpha \approx 4$, which was chosen to represent a pathloss similar to that seen in a typical cellular network.  Note that this pathloss value is indicative of NLOS range measurements, which are inherently a part of localization in cellular networks.  Recall however, that Assumption \ref{Independent_Range} implied the use of Line-of-Sight (LOS) measurements.  Thus, we attempt to mimic NLOS in our simulations by selecting a range error to account for a reasonable delay spread under NLOS conditions.  This is described further in the following section. We note that a subject of future work will be to seek a refinement of this model by incorporating NLOS directly into the range measurements.
	
	The last parameter that remains fixed across simulations is $M$.  This parameter was chosen to be large enough such that an examination of the cdf of $S$ will reveal which percentage of the network the target is unlocalizable.  Towards this end, it was sufficient to choose $M \approx 200 \, m$ for our simulations here.  Again, the choice of this value is arbitrary and left up to network designers and how they would like to treat the unlocalizable cases.

\subsubsection{Conducting the Simulations}

	The true marginal cdf of $S$ was generated through a simulation over 100,000 positioning scenarios.  Each scenario consisted of an average placement of 1,000 anchor nodes, placed according to a homogeneous PPP, with the target located at the origin.  Next, the anchor nodes whose SIRs surpassed the detection threshold were deemed to participate in the localization procedure, and their corresponding coordinates were used to calculate the true value of $S$ given by Definition \ref{DefnOfS}.  If more than $N$ anchor nodes had signals above the threshold, then the $N$ anchors with highest SIRs were used to calculate $S$.

\subsection{The Effect of Frequency Reuse on the Network-Wide Distribution of Localization Error} \label{FreqReuseSection}

	In this section, we explore how frequency reuse impacts localization performance throughout the network.  This simulation, as well as all subsequent simulations, compares the true (simulated) distribution of $S$ with our analytical model from Theorem \ref{MainThrm}.  All parameters were fixed at the levels stated in Fig. \ref{Varying_K}, while the only parameter varied was the frequency reuse factor, $\mathcal{K}$. The range error standard deviation, $\sigma_r$, was chosen according to the detection threshold, $\beta$, and the CRLB of a range estimate (see \cite{Kay}, equ. 3.40), assuming a 10MHz channel bandwidth. Approximately $10m$ was added to account for a reasonable delay spread in NLOS conditions.

\begin{figure}[t]
\centering
\begin{minipage}[t]{0.48\textwidth}
\centering
\includegraphics [scale=0.60]{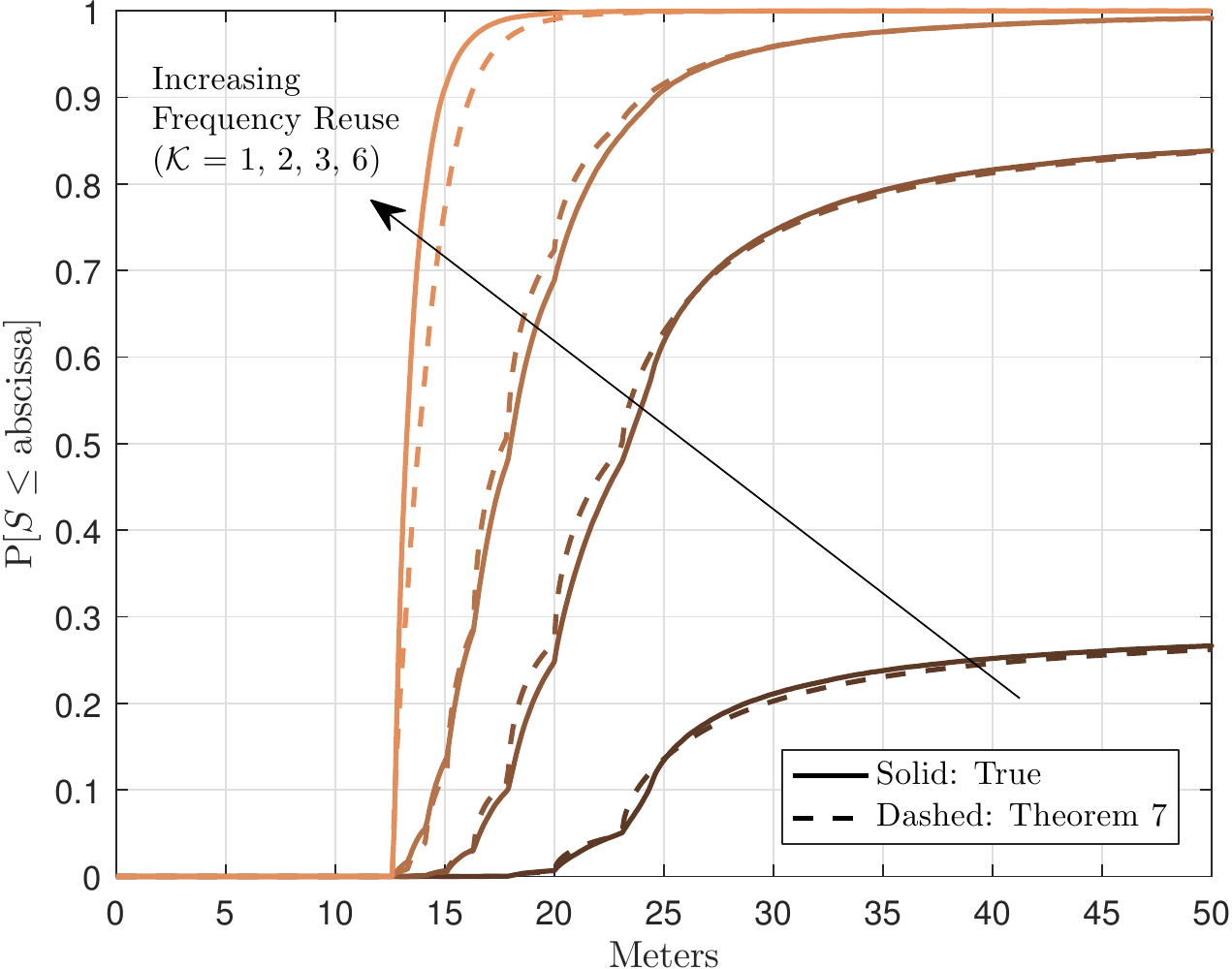}
\caption{\textsc{The Effect of Frequency Reuse}.  This result exposes the large impact that frequency reuse has on localization performance throughout the network.  The parameters are: $N = 10$, $\beta = 10\text{ dB}$, $\gamma = 20\text{ dB}$, $q = 1$, and $\sigma_r = 20 \, m$.  The plots' ``piece-wise'' appearance is due to $L$ being discrete.}
\label{Varying_K}
\end{minipage}\hfill
\begin{minipage}[t]{0.48\textwidth}
\centering
\includegraphics [scale=0.60]{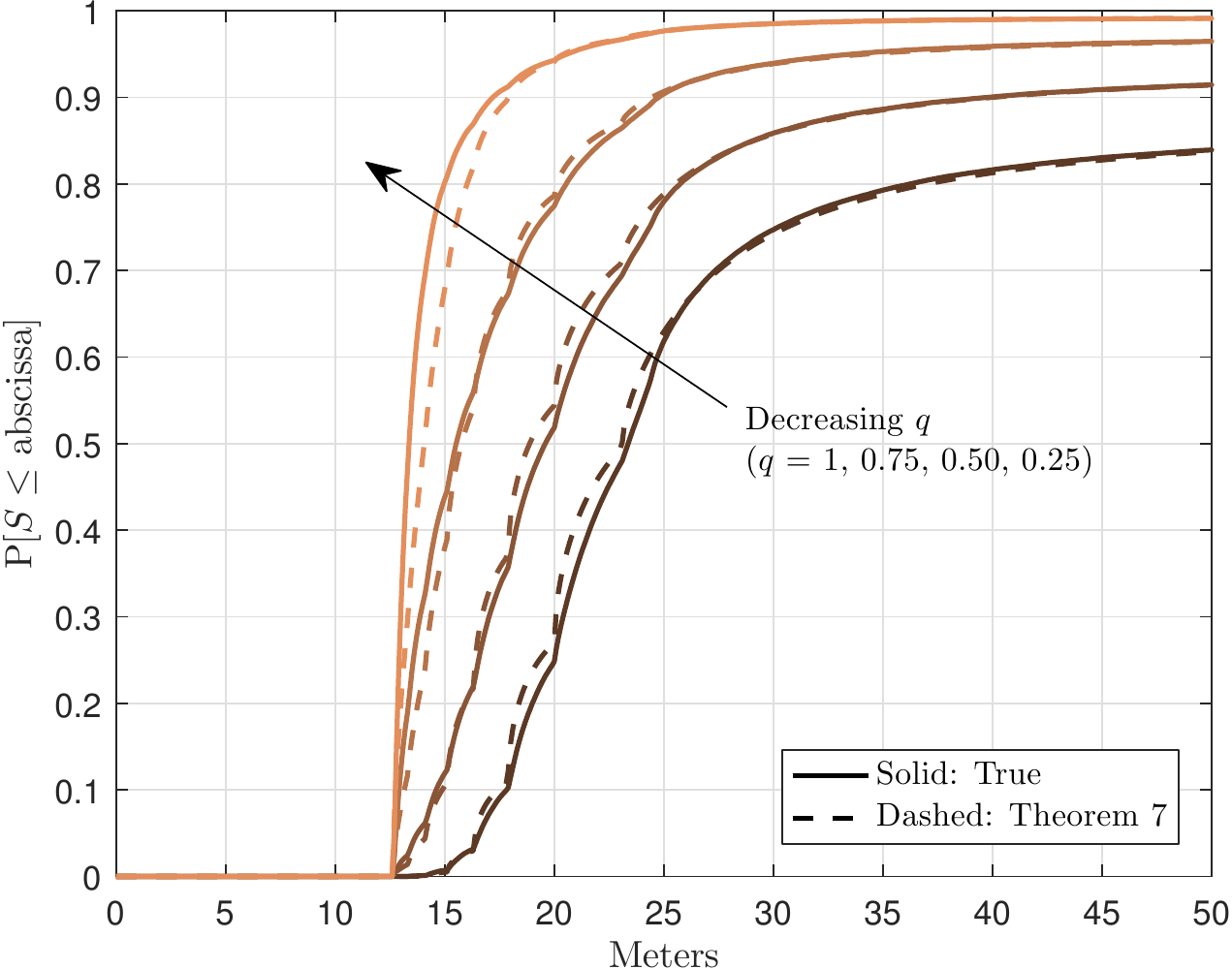}
\caption{\textsc{The Impact of Decreasing Network Load}. This plot demonstrates the improvement in localization performance due to a decrease in network loading.  The parameters were chosen as follows: $N = 10$, $\beta = 10\text{ dB}$, $\gamma = 20\text{ dB}$, $\mathcal{K} = 2$, and $\sigma_r = 20 \, m$.}
\label{Varying_q}
\end{minipage}
\end{figure}

		From Fig. \ref{Varying_K}, the most notable impact that frequency reuse has on localization performance is that of localizability. That is, with just a small increase in frequency reuse from $\mathcal{K} = 1$ to $\mathcal{K} = 2$, the portion of the network with which a target is localizable increases from only $\approx \! 25 \%$ to an astonishing $\approx \! 85 \%$.  Furthermore, localization error is also reduced, although the improvement is not as drastic as the increase in localizability.  Additionally, as frequency reuse increases, the gains in localizability stop after $\mathcal{K} = 3$, with the gains in localization error also declining after $\mathcal{K} = 3$ as well.  Thus, we can conclude that an increase in frequency reuse is strongly advisable if once desires an increase in localization performance within a network, a result which coincides with what has been seen in practice, \emph{viz.} 3GPP.  Lastly, we note the excellent match between the true simulated distribution and our analytically derived distribution given by Theorem \ref{MainThrm}.  We will see that this accuracy of Theorem \ref{MainThrm} is retained across all of our results in this section.

\subsection{Examining the Effects of Network Loading}

	Here we examine the effect that network loading has on localization performance throughout the network.  This is accomplished by varying the percentage of the network, $q$, actively transmitting (interfering) during a localization procedure.  All parameter values other than $q$ were fixed and $\sigma_r$ was chosen in the same manner as in the frequency reuse case.


	From the distributions plotted in Fig. \ref{Varying_q}, we can see that a decrease in network load leads to an improvement in localizability, as well as an improvement in localization error.  However, the improvement is not as pronounced as in the frequency reuse case.  Further, examining the $80^\text{th}$ percentile for example, it is evident that the rate of improvement in localization error declines as the network load declines as well.  Thus, since low network traffic is usually never desirable, a network designer looking to optimize localization performance may find solace in the fact that gains in performance begin to decline as network loading decreases also.  

\subsection{The Impact of Processing Gain}

	In this section, we examine the effects of changing the processing gain, since it is perhaps the easiest parameter for a network designer to change in practice.  (Note that we choose $\sigma_r$ here as in the previous simulations.)
From Fig. \ref{Varying_gamma}, it is evident that as the processing gain increases, there is a corresponding improvement in localizability across the network, as well as an improvement in localization error.  As a consequence, there exists a clear trade-off between sacrificing processing time for gains in localization performance.  However, it appears that these improvements begin to level off at a processing gain of $\approx$25 dB. This is promising, as processing gains higher than this can quickly become impractical.  Examining the $80^\text{th}$ percentile, we can see that a 10 dB increase in processing gain can lead to about a $30m$ improvement in localization error throughout the network. Thus, increasing the processing gain can be an easily implementable solution for achieving moderate gains in localization performance.


\begin{figure}[t]
\centering
\begin{minipage}[t]{0.48\textwidth}
\centering
\includegraphics [scale=0.60]{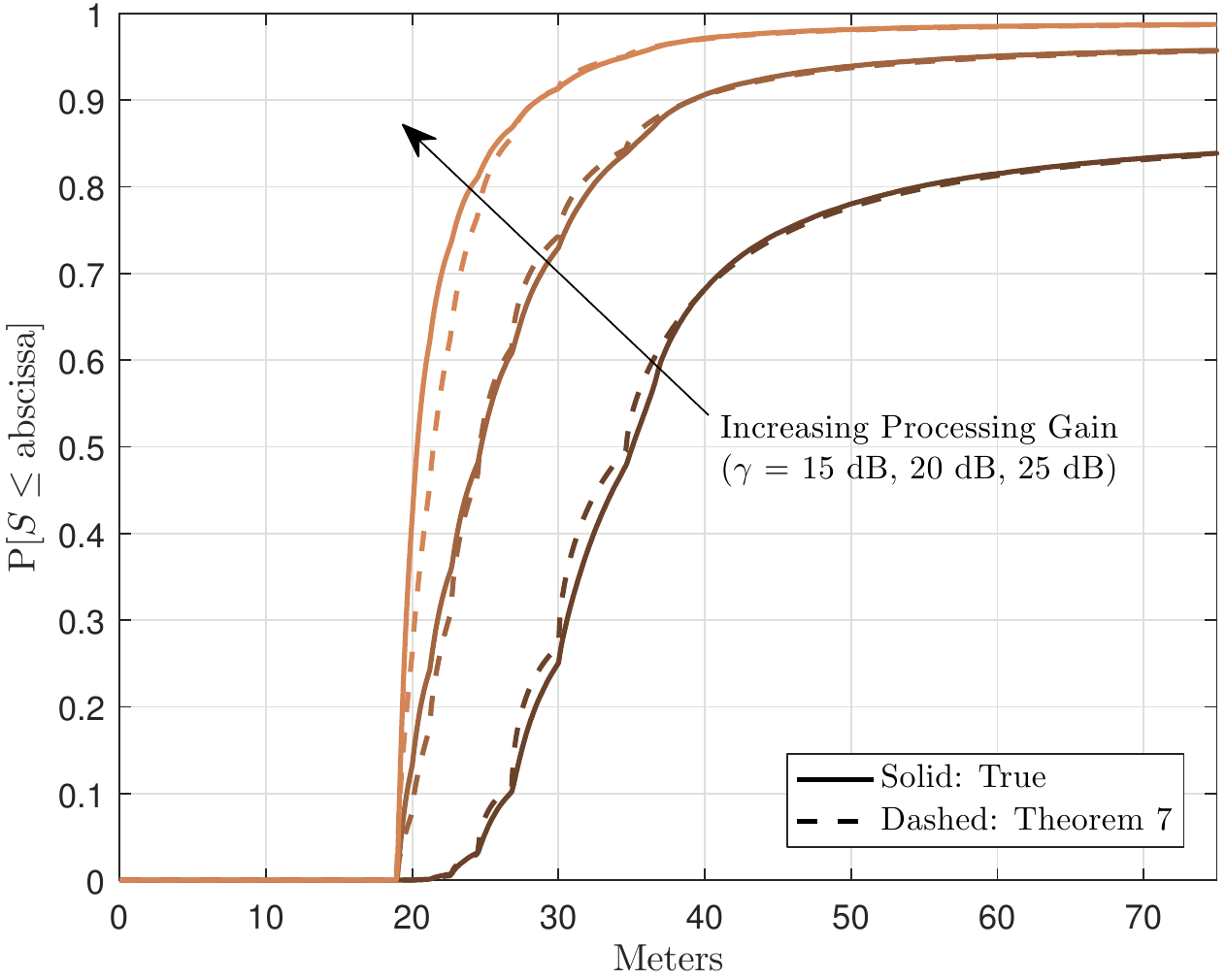}
\caption{\textsc{The Effect of Increasing Processing Gain}. This figure highlights the trade-off that exists between processing time and localization performance.  The distribution parameter values are: $N = 10$, $\beta = 5\text{ dB}$, $q = 1$, $\mathcal{K} = 2$, and $\sigma_r = 30 \, m$.}
\label{Varying_gamma}
\end{minipage}\hfill
\begin{minipage}[t]{0.48\textwidth}
\centering
\includegraphics [scale=0.60]{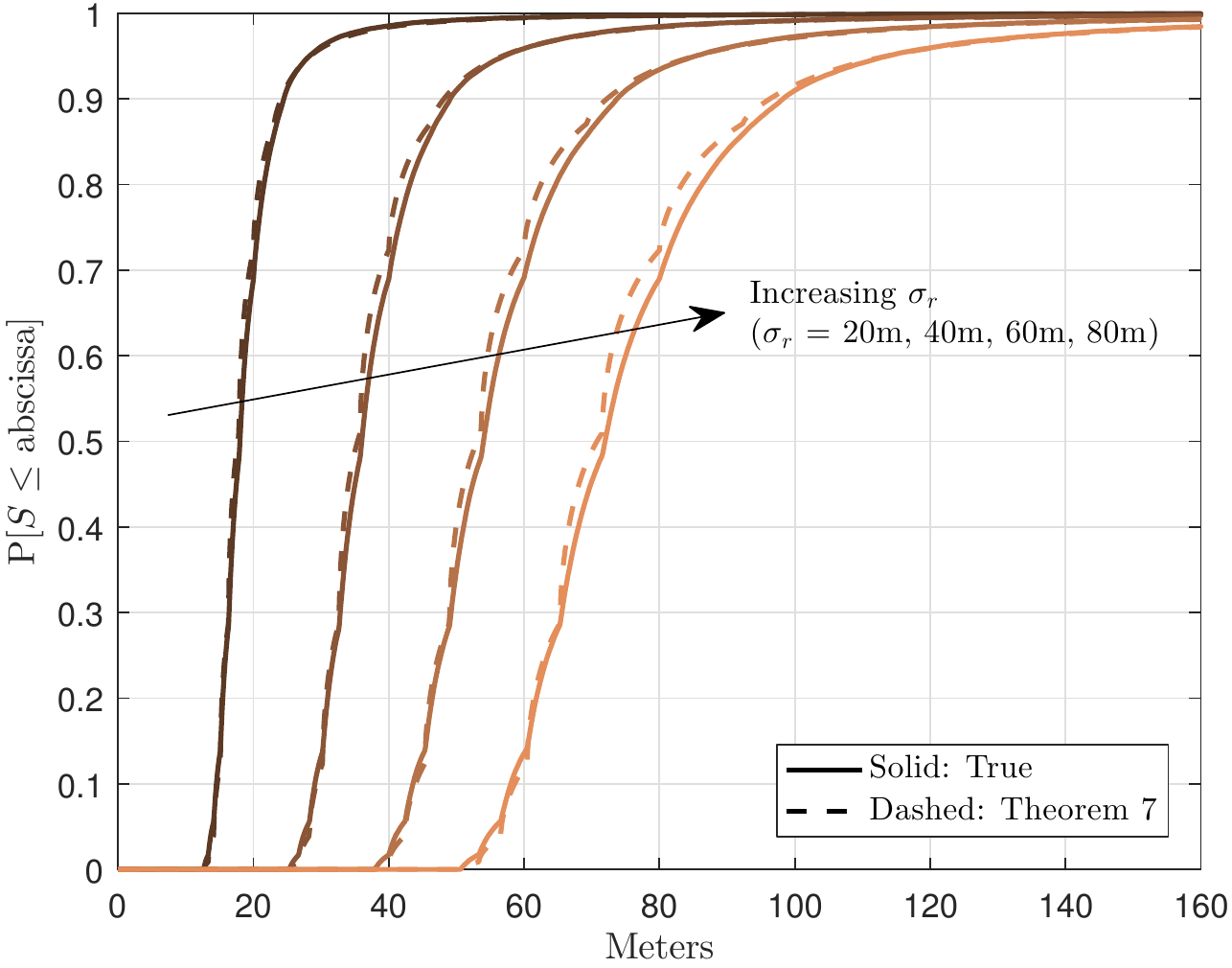}
\caption{\textsc{The Impact of Increasing Range Error}. Injecting range error results in a predictable effect on localization performance throughout the network, yet has no effect on localizability.  The parameter values are: $N = 10$, $\beta = 10\text{ dB}$, $\gamma = 20\text{ dB}$, $\mathcal{K} = 3$, and $q = 1$.}
\label{Varying_sigma}
\end{minipage}
\end{figure}

\subsection{Range Error and its Impact on Localization Performance within the Network} \label{Mimic_NLOS}

	With this last result, we attempt to mimic the effect of an increasing NLOS bias by injecting additional range error into the measurements.  
In examining Fig. \ref{Varying_sigma}, we note that the first $\sigma_r$ value, \emph{i.e.}, $\sigma_r = 20 \, m$, was chosen as in the frequency reuse simulation.  Therefore, the subsequent choices represent that of mimicking the impact of increasing NLOS bias.  From  Fig. \ref{Varying_sigma}, we can see that increasing range error has no effect on localizability within the network.  This should be clear from the analytical model, since $\sigma_r$ does not appear as a parameter in Corollary \ref{CorFreqReuse}.  Additionally, injecting more range error into the measurements results in a predictable effect on the distribution of localization performance.  This is also evidenced by examining Theorem \ref{SgivenL}, where $\sigma_r$ appears as a scale parameter.  Thus, mimicking the effects of NLOS measurements results in a scaling of the distribution of $S$, implying a predictable reduction in localization performance.


\section{Conclusion}

	This paper presents a novel parameterized distribution of localization error, applicable throughout an entire wireless network.
Invoking a PPP network model, as well as common assumptions for TOA localization, has enabled this distribution to simultaneously account for all possible positioning scenarios within a network.  Deriving this result involved two main steps: (1) a derivation of an approximate distribution of our localization performance benchmark, $S$, \emph{conditioned} on the number of hearable anchor nodes, $L$, which yielded a conditional distribution with desirable accuracy and tractability properties, and (2) a modification of results from \cite{Javier_Journal} to attain the distribution of $L$. Thus, using these two distributions, we arrived at the final, marginal distribution of $S$, which also retained the same parameterizations as its two component distributions. 

	What followed was a numerical analysis of this network-wide distribution of localization performance.  This analysis revealed that our distribution can offer an accurate, baseline tool for network designers, in that it can be used to get a sense of localization performance within a wireless network, while also providing insight into which network parameters to change in order to meet localization requirements.

	Since this marginal distribution of $S$ is a distribution of the TOA-CRLB throughout the network, it consequently provides a benchmark for describing localization performance in any network employing an unbiased, efficient, TOA-based localization algorithm.  The amount of insights that this network-wide distribution of localization error reveals are numerous, and the results presented in this paper have only begun to explore this new paradigm.  It is our hope that this work spawns additional research into this new concept, and that future work, such as incorporating NLOS measurements, accounting for collaboration, adding other localization strategies (TDOA/AOA), etc., can further refine the model presented here.
	
	 In closing, this work presents an initial attempt to provide network designers with a tool for analyzing localization performance throughout a network, freeing them from lengthly simulations by offering an accurate, analytical solution.


\appendices

\section{Proof of Proposition \ref{RewritingD}}  \label{PropositionProof}

	It first helps to visualize the $\sin^2(\cdot)$ terms of the sum in (\ref{D_Expression}) on a 2-D grid, where the $i$'s represent the rows and the $j$'s represent the columns. As an example, the case of $L=4$ gives
\begin{align*}
\begin{matrix*}[l]
      &               j=1                        &                   j=2                   &               j=3            &               j=4                \\
i=1  & \phantom{\angle \angle_2}  & \angle_1 \phantom{+\angle_2}  & \angle_1+\angle_2    & \angle_1+\angle_2+\angle_3 \\
i=2  &                                           &                                          & \angle_2                    & \angle_2+\angle_3   \\
i=3  &                                           &                                          &                                &  \angle_3               
\end{matrix*}
\end{align*}
where just the arguments of the $\sin^2(\cdot)$ terms are represented in the grid for clarity. From this arrangement, it is evident that the sum in (\ref{D_Expression}) represents the process of summing each row sequentially, starting at $i=1$.  

	Now, however, we choose to sum the terms diagonally, starting with the lowest diagonal and working our way upward. This yields
\begin{align} \label{DiagonalSum}
D = \sum_{i = 1}^{L-1} \sum_{j = 1}^{L-i} \sin^2 \Bigg( \sum_{k = j}^{j+i-1} \angle_k \Bigg).
\end{align}
Considering the cases $i=1$ and $i=L-1$ separately, we may rewrite (\ref{DiagonalSum}) as
\begin{align*}
D = \sum_{k=1}^{L-1} \sin^2 \big(\angle_k\big) + \sum_{i = 2}^{L-2} \sum_{j = 1}^{L-i} \sin^2 \Bigg( \sum_{k = j}^{j+i-1} \angle_k \Bigg) + \sin^2 \Bigg( \sum_{k = 1}^{L-1} \angle_k \Bigg).
\end{align*}
Next, note that
\begin{align}
\sin^2 \Bigg( \sum_{k = 1}^{L-1} \angle_k \Bigg) &= \sin^2 \Bigg( 2\pi -  \sum_{k = 1}^{L-1} \angle_k \Bigg)
= \sin^2 \big( 2\pi - (\Theta_{(L)} - \Theta_{(1)}) \big) = \sin^2 \big( \angle_L \big),
\end{align}
where the last two equalities follow from Definition \ref{Internodal_Angles}.  Hence, 
\begin{align*} 
D = \sum_{k=1}^{L} \sin^2 \big(\angle_k\big) + \sum_{i = 2}^{L-2} \sum_{j = 1}^{L-i} \sin^2 \Bigg( \sum_{k = j}^{j+i-1} \angle_k \Bigg).
\end{align*}
Lastly, we complete the proof by noting that the first sum may be equivalently expressed by replacing the internodal angles with their order statistics.


\section{Proof of Corollary \ref{ExpectationOfSecondLargest}}  \label{CorProof}

	Using Lemma \ref{Angle_Max2} and our assumption of a finite $L$, the pdf of $\angle_{(L-1)}$ conditioned on $L$ is
\begin{align*}
f_{\angle_{(L-1)}}(\varphi | L) &= \frac{\text{d}}{\text{d} \varphi} F_{\angle_{(L-1)}}(\varphi | L) = \sum_{n=0}^\mathcal{X} (-1)^{n} {L \choose n} \Bigg( \frac{n(n-1)(L-1)}{2\pi} \Bigg) \Bigg(1-\frac{n\varphi}{2\pi}\Bigg)^{L-2},
\end{align*} 
where $\mathcal{X} = \mathrm{min}\big\{L, \floor*{2\pi / \varphi}\big\}$ and the support is $0 \leq \angle_{(L-1)} \leq \pi$, as in Lemma \ref{Angle_Max2}.

	Next, note that the lower summation limit may be rewritten as $n=2$ and that the upper summation limit, $\mathcal{X}$, can be simplified to just $L$ by appending an indicator function to the summand.  This gives the logically equivalent expression:
\begin{align*}
f_{\angle_{(L-1)}}(\varphi | L) = \sum_{n=2}^L (-1)^{n} {L \choose n} \Bigg( \frac{n(n-1)(L-1)}{2\pi} \Bigg) \Bigg(1-\frac{n\varphi}{2\pi}\Bigg)^{L-2} \! \! \mathbf{1}  \Bigg[\varphi \leq \frac{2\pi}{n}\Bigg].
\end{align*} 

Using the above expression for the pdf, the expectation is derived as follows:
\begin{align}
\E[\angle_{(L-1)} | L]  &= \int_0^\pi \! \varphi \, f_{\angle_{(L-1)}}(\varphi | L) \, \text{d} \varphi \nonumber \\
& = \int_0^{\frac{2\pi}{n}} \! \! \! \varphi \sum_{n=2}^L (-1)^{n} {L \choose n} \Bigg( \frac{n(n-1)(L-1)}{2\pi} \Bigg) \Bigg(1-\frac{n\varphi}{2\pi}\Bigg)^{L-2} \text{d} \varphi \label{Cstep2}\\
&= \sum_{n=2}^L (-1)^{n} {L \choose n} \Bigg( \frac{n(n-1)(L-1)}{2\pi} \Bigg)  \int_0^{\frac{2\pi}{n}} \! \! \! \varphi \Bigg(1-\frac{n\varphi}{2\pi}\Bigg)^{L-2} \text{d} \varphi \label{Cstep3}\\
&= \sum_{n=2}^L (-1)^n {L \choose n} \Bigg(\frac{2\pi(n-1)}{nL}\Bigg), \label{Cstep4}
\end{align}
where (\ref{Cstep2}) follows from absorbing the indicator function into the integration limits since $n\geq2$, (\ref{Cstep3}) follows from our assumption of a finite $L$, and (\ref{Cstep4}) is derived through integration by parts.


\section{Proof of Lemma \ref{Dmax}}  \label{DmaxProof}
	This is a straightforward application of the lowest Geometric Dilution Of Precision (GDOP) presented in \cite{LowestGDOP2}.  First, since $L$ is finite, then $D$ is a continuous real function defined on a compact subset of $\mathbb{R}^L$.  Thus, its maximum must exist.  Call it $d_\text{max}$.  
	
	Next, under our Assumptions \ref{2WTOA}, \ref{Independent_Range}, and \ref{Common_Error}, GDOP, presented in \cite{cdf_angle_max}, can be written as $\text{GDOP} = \sqrt{L/D}$,
where $D$ follows from the lemma assumption.  Further, under our assumptions, \cite{LowestGDOP2} asserts that the lowest GDOP is given by $\text{GDOP}_\text{min} = 2/\sqrt{L}$.
Since $\text{GDOP}_\text{min}$ must occur when $D$ is at its maximum, we have that $2/\sqrt{L} = \sqrt{L/d_\text{max}}$, 
and the lemma follows.


\section{Proof of Theorem \ref{SgivenL}}  \label{TheoremProof}

	Under Approximation \ref{approximation1}, we have $S = a/\sin (\angle_{(L-1)})$,
where $a$ is defined in Theorem \ref{SgivenL}.  To determine the support of $S$ conditioned on $L$, we know from Lemma \ref{Angle_Max2} that if $0 \leq \angle_{(L-1)} \leq \pi$, then the support for the RV $\sin ( \angle_{(L-1)} )$ must be $0 \leq \sin (\angle_{(L-1)}) \leq 1$.  From here, we see that
\begin{align*}
0 \leq \sin (\angle_{(L-1)}) \leq 1 ~~~~\Rightarrow~~~~ 1 \leq \frac{1}{\sin (\angle_{(L-1)}) } ~~~~\Rightarrow~~~~ a \leq \frac{a}{\sin (\angle_{(L-1)}) } ~~~~\Rightarrow~~~~ a \leq S,
\end{align*}
and hence, $S \in [a, \infty)$.

Next, to find the cdf of $S$ conditioned on $L$, consider the following
\begin{align}
F_S(s \! \mid \! L, \sigma_r) &= \text{P}[S \leq s \mid L, \sigma_r] \nonumber\\
&=\text{P}\Bigg[ \frac{a}{\sin (\angle_{(L-1)})} \leq s ~\Bigg|~ L \Bigg] \label{follows_from_A1}\\
&=\text{P}\Bigg[ \angle_{(L-1)} \! \leq \! \pi \! - \! \sin^{-1} \! \Bigg( \frac{a}{s} \Bigg) \! ~\Bigg|~ \! L \Bigg] - \text{P}\Bigg[ \angle_{(L-1)} \! \leq \sin^{-1} \! \Bigg( \frac{a}{s} \Bigg) \! ~\Bigg|~ \! L \Bigg] \nonumber\\
&=\text{P}[ \angle_{(L-1)} \leq \varphi_2 \mid L ] - \text{P}[ \angle_{(L-1)} \leq \varphi_1  \mid L ] \label{varphi1_varphi2}\\
&=F_{\angle_{(L-1)}} \big(\varphi_2 \mid L \big) - F_{\angle_{(L-1)}} \big(\varphi_1 \mid L \big), \label{last_line}
\end{align}
where (\ref{follows_from_A1}) follows from Approximation \ref{approximation1} and the fact that we may drop the parameter, $\sigma_r$, as a condition since the dependency is now explicit.  Further, $\varphi_1$ and $\varphi_2$ from (\ref{varphi1_varphi2}) are defined in Theorem \ref{SgivenL}, and (\ref{last_line}), through use of Lemma \ref{Angle_Max2}, gives (\ref{SgivenL_cdf}) from Theorem \ref{SgivenL}, as desired.


\end{document}